%% file: main.tex
\documentclass[11pt]{article}
\input{packages.tex}
\input{Adam.tex}
\input{headers}

\title{Quantum Event Learning and Gentle Random Measurements}
\author{Adam Bene Watts\thanks{Institute for Quantum Computing, University of Waterloo, \texttt{adam.benewatts1@uwaterloo.ca}} \and John Bostanci\thanks{Columbia University, \texttt{johnb@cs.columbia.edu}}}

\begin{document}

\maketitle

\begin{abstract}
    We prove the expected disturbance caused to a quantum system by a sequence of randomly ordered two-outcome projective measurements is upper bounded by the square root of the probability that at least one measurement in the sequence accepts. We call this bound the \textit{Gentle Random Measurement Lemma}. 
    
    We then consider problems in which we are given sample access to an unknown state $\rho$ and asked to estimate properties of the accepting probabilities $\Tr[M_i \rho]$ of a set of measurements $\{M_1, M_2, \ldots , M_m\}$. We call these types of problems \textit{Quantum Event Learning Problems}. Using the gentle random measurement lemma, we show randomly ordering projective measurements solves the Quantum OR problem, answering an open question of Aaronson. We also give a Quantum OR protocol which works on non-projective measurements but which requires a more complicated type of measurement, which we call a \textit{Blended Measurement}. Given additional guarantees on the set of measurements $\{M_1, \ldots, M_m\}$, we show the Quantum OR protocols developed in this paper can also be used to find a measurement $M_i$ such that $\Tr[M_i \rho]$ is large. We also give a blended measurement based protocol for estimating the average accepting probability of a set of measurements on an unknown state.

    Finally we consider the \textit{Threshold Search Problem} described by O'Donnell and B\u{a}descu. By building on our Quantum Event Finding result we show that randomly ordered (or blended) measurements can be used to solve this problem using $O(\log^2(m) / \epsilon^2)$ copies of $\rho$. Consequently, we obtain an algorithm for Shadow Tomography which requires $\tilde{O}(\log^2(m)\log(d)/\epsilon^4)$ samples, matching the current best known sample complexity. This algorithm does not require injected noise in the quantum measurements, but does require measurements to be made in a random order and so is no longer online. 
\end{abstract}

\newpage
\tableofcontents

\newpage
\section{Introduction}
\label{sec:Introduction}
Quantum measurements change the states that they act on, often in undesired ways. Colloquially called the ``information-disturbance trade-off'', the \emph{Gentle Measurement Lemma} bounds the damage that a single measurement can cause to a quantum system by relating the probability of a particular outcome to the disturbance when seeing that outcome~\cite{aaronson2004limitations,winter1999coding}. 
\begin{lem*}[Gentle Measurement Lemma, informal] Let $\rho$ be a quantum state, $M$ be a two-outcome measurement and $\rho'$
be the post measurement state when the reject outcome is observed. Then 
\begin{align*}
    \norm{\rho - \rho'}_1 \leq 2 \sqrt{\Tr[M\rho]}\,. 
\end{align*}
\end{lem*}

Notably, the Gentle Measurement Lemma only bounds the disturbance caused by a single measurement.
The \emph{Anti-Zeno Effect} refers to a phenomenon in which a sequence of two-outcome measurements can cause arbitrarily large damage to quantum system, despite the probability of \textit{any} measurement in the sequence accepting being arbitrary small~\cite{kaulakys1997quantum}.\footnote{Of course, one can apply the Gentle Measurement Lemma repeatedly to each measurement in a sequence of measurements. The reason this approach does not rule out phenomenon such as the Anti-Zeno effect comes from the square root in the original gentle measurement lemma -- even when the sum over all measurements' accepting probabilities is small, the sum of the square roots of their accepting probabilities, and hence the resulting bound coming from sequential applications of the gentle measurement lemma, can be large.} For sequential measurements, the closest analogue we know to the Gentle Measurement Lemma is known as the Gentle Sequential Measurement Lemma~\cite{gao2015quantum} (closely related to the Quantum Union Bound~\cite{gao2015quantum,khabbazi2019union}).

\begin{lem*}[Gentle Sequential Measurement, informal]
   Let $\rho$ be a quantum state and $(M_1, M_2, \ldots , M_m)$ be a sequence of $m$ two-outcome measurements. Let $\rho'$
   be the post measurement state of a quantum system, initially in state $\rho$, if measurements $(M_1, M_2, \ldots , M_m)$ are applied in sequence to the system and all reject. Then 
   \begin{align*}
       \norm{\rho - \rho'}_1 \leq 2 \sqrt{\sum_{i} \Tr[M_i \rho]}\,.
   \end{align*}
\end{lem*}

Crucially, the Gentle Sequential Measurement Lemma bounds the damage a sequence of measurements can cause to a system in terms of the accepting probability of each measurement on the \textit{initial} state of the system, not the accepting probability of the measurements on the state on which they are applied. Thus, the Gentle Sequential Measurement Lemma does not rule out phenomenon such as the Anti-Zeno Effect. 

The analysis of sequential measurements is closely related to a class of problems we call \emph{Event Learning Problems}. These problems involve an unknown state $\rho$ and set of measurements $M_1, M_2,\ldots, M_m$. The goal is to learn properties of the measurements' accepting probabilities $\Tr[M_1 \rho], \Tr[M_2 \rho],\ldots, \Tr[M_m \rho]$, while using as few copies of the quantum state $\rho$ as possible. 
This class of problems includes the well studied \emph{Shadow Tomography} problem~\cite{aaronson2019shadow,huang2020predicting,aaronson2019gentle}, but also ``easier'' problems where the goal is to learn fewer features of the accepting probabilities.
Another well studied Event Learning problem is the \emph{Quantum OR} problem. Here, the goal is to approximate the OR of the measurement accepting probabilities, or, more formally, to distinguish between the following cases:
\begin{enumerate}[(1)]
    \item There exists a measurement $M_i$ which accepts on $\rho$ with high probability. \label{item:or_case_accept}
    \item The total accepting probability of all measurements $\sum_i \Tr[M_i \rho]$ is small. 
    \label{item:or_case_reject}
\end{enumerate}

The Quantum OR problem serves as an illustrative example of how the Anti-Zeno Effect, and in general the information-disturbance trade-off, represents substantial barrier to obtaining algorithms for these kinds of tasks that have low sample complexity.  In Ref.~\cite{aaronson2006qma}, Aaronson proposed an algorithm for the Quantum OR problem in which a system was prepared in state $\rho$ and measurements $M_1, \ldots, M_m$ were applied to the system in an order chosen uniformly at random. The claim was that in Case \labelcref{item:or_case_accept} a measurement would eventually accept on the system with reasonably high probability, while in Case \labelcref{item:or_case_reject} no measurement should accept. Ref.~\cite{harrow2017sequential} pointed out a gap in Aaronson's analysis that was closely related to the Anti-Zeno Effect: the original argument did not rule out the possibility that in Case \labelcref{item:or_case_accept} it might be possible that with high probability over the random choice of sequence, the measurements had an Anti-Zeno effect, i.e. all of the measurements in the sequence could all reject with high probability while still causing a large disturbance to the system initially state $\rho$. Ultimately, this disturbance could cause measurement $M_i$ to reject with high probability, despite it accepting with high probability on the initial state.

The authors of Ref.~\cite{harrow2017sequential} gave alternate algorithms which solved the Quantum OR problem. These algorithms still required only a single copy of $\rho$, but involve more complicated measurements than Aaronson's original proposal, and cast aside the idea of using randomly ordered sequences of measurements to solve the problem. 
Despite the gap found in Aaronson's analysis, no counterexample or proof was given, and it remained open whether randomly ordered measurements could solve the Quantum OR problem.  In this paper, we create a toolkit for analyzing the effects of randomly ordered measurements on unknown quantum states.  Using these tools, we are able to show that (a small modification to) the original Aaronson Quantum OR algorithm does indeed work, and simplify algorithms that achieve the best known sample complexity on other event learning tasks, most notably shadow tomography, using random sequences of measurements.  

\subsection{Results}
The first major result in this paper is a generalization of the Gentle Measurement Lemma to the setting where a random sequence of measurements is applied to a state $\rho$. Like the Gentle Measurement Lemma, this lemma gives an upper bound on the ``damage'' (in trace distance) this sequence of measurements can cause to the state $\rho$ in terms of the probability that at least one measurement in the sequence accepts. 



\begin{thm}[Gentle Random Measurement Lemma] \label{thm:random_gentle_total_accept}
Let $\cM = \{M_1, M_2, \ldots, M_m\}$ be a set of two outcome projective measurements, and $\rho$ be a density matrix. Consider the process where a measurement from the set $\cM$ is selected universally at random and applied to a quantum system initially in state $\rho^{(0)} = \rho$. Let $\rho^{(k)}$ be the state of the quantum system after $k$ repetitions of this process where no measurement accepts, so 
\begin{align}
\rho^{(k)} = \frac{\expec_{X_1, \ldots, X_k \sim \cM} \left[ (1-X_k) \ldots (1 - X_1) \rho (1 - X_1) \ldots (1-X_k) \right]}{\expec_{X_1, \ldots, X_k \sim \cM} \left[\Tr\left[ (1-X_k) \ldots (1 - X_1) \rho (1 - X_1) \ldots (1-X_k) \right]\right] }
\end{align}
and let $\accept{k}$ be the probability that at least one measurement accepts during $k$ repetitions of this process (equivalently, the probability that not all measurements reject), so 
\begin{align}
    \accept{k} = 1 - \expec_{X_1, \ldots, X_k \sim \cM} \left[\Tr[ (1-X_k) \ldots (1 - X_1) \rho (1 - X_1) \ldots (1-X_k)]\right]
\end{align}
Then
\begin{align}
\norm{\rho - \state{k}}_1 \leq 4\sqrt{\accept{\lceil k/2 \rceil}} \leq 4\sqrt{ \accept{k}}\,.
\end{align}
\end{thm}
\noindent This theorem shows that randomly ordered sequences of measurements are ``gentle'' in expectation, provided the expectation is taken over all possible orderings. As a consequence, we find that phenomenon similar to the Anti-Zeno Effect are not likely to occur in randomly ordered measurements. \Cref{sec:technicals,sec:blended_measurements} of this paper develops some key ideas which are used in the proof of \Cref{thm:random_gentle_total_accept} (an overview of these ideas is given in \Cref{subsec:methods}). \Cref{sec:random_measurements} proves this theorem. 

In the later half of this paper we use the techniques used to prove \Cref{thm:random_gentle_total_accept} to study several Event Learning problems. In \Cref{sec:Quantum_OR} we consider the Quantum OR problem, and prove correctness of Aaronson's original Quantum OR algorithm, resolving the last unanswered question from Ref.~\cite{aaronson2006qma}.
\begin{thm}[Random Measurements Solve Quantum OR]
Let $\cM = \{M_1, M_2, \ldots , M_m\}$ be a set of two outcome projective measurements. Let $\rho$ be a state such that either there exists an $i \in [m]$ with $\Tr[M_i \rho] > 1- \epsilon$ (Case 1) or $\sum_i \Tr[M_i \rho] \leq \delta$ (Case 2). Then consider the process where $m$ measurements are chosen (with replacement) at random from $\cM$ and applied in sequence to a quantum system initially in state $\rho$: in Case 1, some measurement in the sequence accepts with probability at least $(1-\epsilon)^2/4.5$; in Case 2, the probability of any measurement accepting is at most $\textbf{2}\delta$.
\end{thm}

\noindent We also give a Quantum OR procedure which performs better than the random measurement procedure, and which can be used when the measurements $M_1, \ldots, M_m$ are not projective. However this procedure requires more complicated measurements than the random measurement procedure above. 

\begin{thm} [Improved Quantum Or Procedure]\label{thm:blended_quantum_or}
Let $\cM = \{M_1, M_2, \ldots , M_m\}$ be a set of two outcome measurements. Let $\rho$ be a state such that either there exists an $i \in [m]$ with $\Tr[M_i \rho] > 1- \epsilon$ (Case 1) or $\sum_i \Tr[M_i \rho] \leq \delta$ (Case 2). Then there is a test which uses one copy of $\rho$ and in Case 1, accepts with probability at least $(1-\epsilon)^2/4$; in Case 2, accepts with probability at most $\delta$.
\end{thm}

\noindent Details of this test are given in \Cref{subsec:Blend_Measurements_OR}.

In \Cref{sec:quantum_event_finding} we introduce the problem of \textit{Quantum Event Finding}. This is a variant of the Quantum OR problem where the goal is to accept or reject as in Quantum OR and additionally, in the accepting case, to return a measurement $M_i$ such that $\Tr[M_i \rho]$ is large.  We then show the Quantum OR procedure introduced in \Cref{subsec:Blend_Measurements_OR} can be extended to solve Quantum Event Finding in the case when the total weight of undesirable events is bounded by a constant.  
\begin{thm}[Single Copy Event Finding]
Let $\cM = \{M_1, M_2, \ldots , M_m\}$ be a set of two outcome measurements. Let $\rho$ be a state such that either there exists an $i \in [m]$ with $\Tr[M_i \rho] > 1- \epsilon$ (Case 1) or $\sum_i \Tr[M_i \rho] \leq \delta$ (Case 2), and let 
\begin{equation}
    \beta = \sum_{i: \Tr{M_i \rho} \leq 1 - \epsilon} \Tr[M_i \rho]\,.
\end{equation}
Then there is a test which uses one copy of $\rho$, and in Case 1 accepts and outputs a measurement $M_i$ such that $\Tr[M_i \rho] \geq 1-\epsilon$ with probability $(1-\epsilon)^3/(12(1+\beta))$, and in Case 2, accepts with probability at most $\delta$.  
\end{thm}

\noindent Combining this with techniques from \Cref{sec:random_measurements}, we also show that an algorithm similar to Aaronson's original Quantum OR algorithm halts on a desirable measurement with constant probability (again provided the total weight of undesirable events is bounded by a constant).
\begin{thm}[Random Measurements Solve Event Finding]
Let $\cM = \{M_1, M_2, \ldots , M_m\}$ be a set of two outcome measurements. Let $\rho$ be a state such that either there exists an $i \in [m]$ with $\Tr[M_i \rho] > 1- \epsilon$ (Case 1) or $\sum_i \Tr[M_i \rho] \leq \delta$ (Case 2), and let 
\begin{equation}
    \beta = \sum_{i: \Tr{M_i \rho} \leq 1 - \epsilon} \Tr[M_i \rho]\,.
\end{equation}
Then if measurements are chosen uniformly at random (with replacement), in Case 1, with probability at least $(1-\epsilon)^7/(1296(1+\beta)^3)$, at least one measurement accepts and the first accepting measurement satisfies $\Tr[M_i \rho] \geq 1-\epsilon$.  In Case 2, a measurement accepts with probability at most $2\delta$. 
\end{thm}

In \Cref{sec:threshold_search}, we consider the similar problem of \emph{Quantum Threshold Search}, introduced in Ref.~\cite{badescu2021improved}.  We begin by extending the event finding results to show that independent of $\beta$, the distribution over measurements output by the procedure is correlated with the relative magnitude of $\Tr[M_i \rho]$.  We leverage this to state a novel threshold search algorithm based on blended (and random) measurements and prove that it requires $O(\log^2 m)$ copies of the unknown state $\rho$, matching the best known upper bound. 
\begin{thm}[Random Measurement Threshold Search]
    Let $\cM = \{M_1, M_2, \ldots, M_m\}$ be a set of two-outcome measurements and $\rho$ be an unknown quantum state.  
    Consider the process where a uniformly random threshold $\theta \in [2/5, 3/5]$ is chosen, then $m$ measurements are chosen (with replacement) at random from $\cM$, and the corresponding binomial measurements (with threshold $\theta$) are applied in sequence to a quantum state initially in $\rho^{\otimes O(\log^2(m))}$, halting if any measurement accepts.  If there is a measurement in $\cM$ satisfying $\Tr[M_i \rho] \geq 3/4$, then this procedure halts on a measurement satisfying $\Tr[M_i \rho] \geq 1/3$ with constant probability.  
\end{thm}

A series of reductions from Ref.~\cite{badescu2021improved} show that this implies a shadow tomography algorithm that uses $\widetilde{O}(\log(d) \log^2(m) / \epsilon^4)$ copies of $\rho$.  We believe that while this algorithm is no longer online, as it requires measurements to be made in random order, it represents a promising path towards improving the upper bounds on shadow tomography, which we discuss in \cref{sec:shadow_tomography_commentary}.

Finally, in \Cref{sec:quantum_mean_estimation} we introduce the problem of \textit{Quantum Mean Estimation}. In this problem our goal, given a description of measurements $M_1, \ldots, M_m$ and sample access to a state $\rho$, is to estimate the average accepting probability of the measurements: $\frac{1}{m} \sum_i \Tr[M_i \rho]$. We give a natural algorithm for this problem, and analyze its performance. This analysis show the surprising result that the average of non-commuting measurements (i.e. $M_1 = \ketbra{1}, M_2 = \ketbra{+}$) can sometimes be estimated to fixed accuracy using \textit{fewer} samples than would be required to estimate the mean of classical measurements. 

\subsection{Techniques}
\label{subsec:methods}

The key technique used in this paper to prove results about sequences of random measurements is a two step procedure where we: (1) introduce and analyze an easier-to-study sequence of measurements, which we call \emph{blended measurements}, then; (2) relate the behavior of a quantum system after a sequence of blended measurements to its behavior after a sequence of random measurements. 


Blended measurements are discussed at length in \Cref{sec:blended_measurements}. Here we give their definition and discuss a few of their basic properties. Given a set of two outcome measurements $\cM = \{M_1, M_2, \ldots, M_m\}$ the blended measurement $\cB(\cM)$ is defined to be the $m+1$ outcome measurement with measurement operators 
\begin{align}
E_0 &= \sqrt{1 - \sum_{i=1}^m \frac{M_i}{m}} \;\;\; \text{ and }\\
E_i &= \sqrt{\frac{M_i}{m}}\,.
\end{align}
We refer to outcome $E_0$ as the ``reject'' outcome, and outcomes $E_1, \ldots, E_m$ as the ``accept'' outcomes. In \Cref{sec:blended_measurements} we show that repeated blended measurements naturally satisfy a gentle measurement lemma, in the sense that the disturbance to a state after $k$ blended measurements is bounded by the (square root of) the probability that at least one measurement accepts. The proof of  the Gentle Random Measurement Lemma (\Cref{thm:random_gentle_total_accept}) starts with this observation and then relates the blended and random measurement procedures using a technical lemma based on the Cauchy-Schwarz inequality (\Cref{lem:taking_avg_outside_trace}, discussed in the next section).


\subsection*{Acknowledgements}
The authors would like to that Aram Harrow, Scott Aaronson and Luke Schaeffer for helpful discussions. Particular thanks is due to Luke Schaeffer for finding the counterexample discussed in Appendix B of this paper.  The authors would also like to thank Matthias Caro for spotting a typo in an earlier version of this paper. Additionally, ABW would like to thank Nilin Abrahamsen and Juspreet Singh Sandhu for a particularly motivating discussion over ice cream during the early days of this work.  JB would like to thank Ashwin Nayak and Angus Lowe for the many useful discussions and motivation.  JB is partially supported by NSF CAREER award CCF-2144219. 

\section{Notation and Preliminaries}

A quantum system $\system{R}$ (as indicated by the font) is a named  finite dimensional complex Hilbert space.  Given two quantum systems $\system{A}$ and $\system{B}$, denote by $\system{AB}$ the tensor product of the two associated Hilbert spaces.  For a linear operator $L$ acting on system $\system{R}$, we sometimes use the notation $L_{\system{R}}$ to indicate that $L$ acts on system $\system{R}$, and we similarly use  $\rho_{\system{R}}$ to denote that $\rho$ is a state in the quantum system $\system{R}$.  When clear from context, we drop the system subscript.  We write $\Tr[\cdot]$ to mean the trace, and $\Tr_{\system{R}}[\cdot]$ to mean the partial trace over system $\system{R}$. 

Given a linear operator $X$, define $\norm{X}_1 = \Tr(|X|)$ to be its trace norm.  For two linear operators $X$ and $Y$ acting on the same system, we say that $X \leq Y$ is $Y - X$ is a positive semi-definite operator.  For two quantum states $\rho$ and $\sigma$, we define the fidelity $\fidelity{\rho}{\sigma} = \norm{\sqrt{\sigma}\sqrt{\rho}}_1$.  We denote by $1$ the identity operator, where the system it acts on should be apparent from context.  

A quantum measurement is defined by a finite set of positive semi-definite operators $\{\sqrt{M_i}\}_{i}$ acting on the same quantum system, satisfying $\sum_{i} M_i = 1$, where each $\sqrt{M_i}$ is associated with outcome $i$.  When there are $2$ outcomes in a measurement, described by the matrices $\{\sqrt{M}, \sqrt{1 - M}\}$, we refer to the resulting measurement as ``the two-outcome measurement $M$''.  For two-outcome measurements, we refer to the $\sqrt{M}$ outcome as the ``accepting'' outcome, and the $\sqrt{1 - M}$ outcome as the ``rejecting'' outcome.  We now state the Gentle Measurement Lemma formally.

\begin{lem}[Gentle Measurement Lemma] Let $\rho$ be a quantum state and $0 \leq M \leq 1$ be a two-outcome measurement. Let $\epsilon := \Tr[M\rho]$ be the accepting probability of the measurement on a quantum system in state $\rho$ and
\begin{align}
    \rho' := \frac{\sqrt{1-M} \rho \sqrt{1-M}}{\Tr[(1-M)\rho]}\,. 
\end{align}
be the post measurement state when the reject outcome is observed. Then 
\begin{align}
    \norm{\rho - \rho'}_1 \leq 2 \sqrt{\epsilon}\,. 
\end{align}
\end{lem}
Throughout the paper we use the following inner product on Hermitian operators, called the Hilbert-Schmidt inner product, defined as follows
\begin{equation}
    \ip{A}{B} = \Tr[AB^{\dagger}]\,.
\end{equation}
An important fact about inner products, including the Hilbert-Schmidt inner product, is the Cauchy-Schwarz inequality, defined below
\begin{equation}
    \ip{A}{B} \leq \sqrt{\ip{A}{A}\ip{B}{B}}\,.
\end{equation}
Another important fact that will appear in the paper is the arithmetic-geometric mean inequality.  Given two non-negative real numbers $a$ and $b$, the following holds
\begin{equation}
    \frac{a + b}{2} \geq \sqrt{ab}\,.
\end{equation}

\section{Gentle Measurement Lemmas}
In this section, we prove that blended and random sequences of measurements obey a variant of the gentle measurement lemma.  These results will be core mathematical tools in showing that algorithms presented in later sections work.
\subsection{Technical Lemmas}\label{sec:technicals}
Many of the results in this paper are derived from general statements about positive semi-definite (PSD) matrices.  In particular, the following lemma, along with the definition of blended measurements, are the core ingredients in our quantum OR results.  

\begin{lem} \label{lem:taking_avg_outside_trace}
   Let $X$ and $Y$ be PSD matrices a $\cA = \{A_1, A_2, \ldots, A_m\}$ be an arbitrary set of matrices and 
   $\{p_1, p_2 \ldots , p_m\}$ be a set of real numbers with $p_i \geq 0$ for all $i$ and $\sum_i p_i = 1$. Then 
   \begin{align}
   \sum_{i,j \in [m]} p_ip_j  \Tr[X  A_i Y  A_j\adj] \leq \sum_{i \in [m]} p_i \Tr[X A_i Y A_i\adj]\,.
   \end{align}
\end{lem}

\begin{proof}
The result follows from Cauchy-Schwarz applied to the Hilbert-Schmidt inner product. We compute:
\begin{align}
   \sum_{i,j \in [m]} p_ip_j  \Tr[X  A_i Y  A_j\adj] &= \sum_{i,j \in [m]} p_ip_j  \ip{\sqrt{Y} A_i\adj \sqrt{X}}{\sqrt{Y}A_j\adj \sqrt{X}}\\ 
   &\leq \sum_{i,j \in [m]} p_ip_j  \sqrt{\ip{\sqrt{Y} A_i\adj \sqrt{X}}{\sqrt{Y} A_i\adj \sqrt{X}}
   \ip{\sqrt{Y}A_j\adj \sqrt{X}}{\sqrt{Y}A_j\adj \sqrt{X}}}\\ 
   &\leq \sum_{i,j \in [m]}  \frac{p_ip_j}{2} \left(\ip{\sqrt{Y} A_i\adj \sqrt{X}}{\sqrt{Y} A_i\adj \sqrt{X}} +  \ip{\sqrt{Y}A_j\adj \sqrt{X}}{\sqrt{Y}A_j\adj \sqrt{X}}\right) \\ 
   &= \sum_i p_i \ip{\sqrt{Y} A_i\adj \sqrt{X}}{\sqrt{Y} A_i\adj \sqrt{X}} \\
   &= \sum_i p_i \Tr[X A_i Y A_i\adj]\,,
\end{align}
where we used Cauchy-Schwarz on the second line, the arithmatic-geometric mean inequality on the third, and the fact that $\sum_{i}p_i = 1$ on the fourth. 
\end{proof}

 We can apply this to prove a corollary more suited to the randomized measurement setting. Before this, we introduce some notation useful for keeping track of the matrix products that appear when analyzing random and blended measurements. 

\begin{defn} \label{defn:multiMeasurement}
Given a set of matrices $\cA = \{A_1, A_2, \ldots , A_m\}$ define the set of matrix products
\begin{align}
    \mm{\cA}{k} = \left\{ \prod_{\alpha = 1}^k A_{i_\alpha} \right\}_{\vec{i} \in [m]^k}\,.
\end{align}
where use the notation $\vec{i} = (i_1, i_2, \ldots, i_k)$ to label components of a vector $\vec{i}$.  $\mm{A}{k}$ contains possible length $k$ products of matrices drawn with replacement from the set $\cA$.  
\end{defn}

\begin{cor}
\label{cor:taking_avg_outside_trace_specific}
Let $\rho$ be a state,  $X$ be a PSD matrix and $\cM$ be a set of $m$ self-adjoint matrices. Define $\mm{\cM}{k}$ as in \Cref{defn:multiMeasurement}. Then 
\begin{align}
    m^{-k} \sum_{T \in \mm{\cM}{k}} \Tr[X T \rho T\adj] &\geq m^{-2k} \sum_{T, S \in \mm{\cM}{k}} \Tr[X T \rho S\adj] = m^{-2k} \sum_{T, S \in \mm{\cM}{k}} \Tr[X T \rho S]\,.
\end{align}
\end{cor}

\begin{proof}
The first inequality is immediate from \Cref{lem:taking_avg_outside_trace} with $Y = \rho$, $\cA = \mm{\cM}{k}$ and $p_1 = p_2 = \ldots = p_{m^k} = m^{-k}$. The second equality holds because 
\begin{align}
    \left(\mm{\cM}{k}\right)\adj =
    \left\{ 
    \left(
    \prod_{\alpha = 1}^k M_{i_\alpha} \right)
    \adj\right\}_{\vec{i} \in [m]^k} 
    =
    \left\{ 
    \prod_{\alpha = 1}^k M_{i_\alpha'} 
    \right\}_{\vec{i}' \in [m]^k} 
    = \mm{\cM}{k}\,.
\end{align}
\end{proof}

\subsection{Gentle Blended Measurements}
\label{sec:blended_measurements}

In this section we prove a number of results about repeated blended measurements.  We begin by repeating the definition of a blended measurement from \Cref{sec:Introduction}.  

\begin{defn}[Blended Measurement]\label{defn:blended_measurement}

Given a set of two outcome measurements $\cM = \{M_1, M_2, \ldots , M_m\}$ the blended measurement $\blend(\cM)$ is defined to be the $m+1$ outcome measurement with measurement operators 
\begin{align}
E_0 &= \sqrt{1 - \sum_{i=1}^m \frac{M_i}{m}} \quad\text{and}\\
E_i &= \sqrt{\frac{M_i}{m}} \quad \text{for $i \in \{1, \ldots, m\}$}\,.
\end{align}
We refer to outcome $E_0$ as the ``reject'' outcome, and outcomes $E_1, \ldots, E_m$ as ``accepting'' outcomes. 
\end{defn}

Intuitively, one might understand a blended measurement $\blend(\cM)$ on a state $\rho$ as the measurement implemented by the following procedure, although the post-measurement state of this procedure is not the same: 
\begin{framed}
\noindent $\cB(\cM)$, intuitively:
\begin{enumerate}
    \item Select a measurement $M_i \in \cM$ universally at random and apply it to $\rho$. 
    \begin{enumerate}
        \item If the measurement accepts, report outcome ``$\blend(\cM)$ accepts on measurement $M_i$''.
        \item If the measurement rejects, report outcome ``$\blend(\cM)$ rejects''. 
    \end{enumerate}
\end{enumerate}
\end{framed}
Critically, the observer performing a blended measurement only learns which measurement $M_i$ was selected if that measurement accepts; otherwise all they learn is that a measurement rejected, without knowing which measurement was performed. 
We will be particularly interested in the analyzing what happens when $k$ blended measurements are applied in sequence to a quantum system initially in state $\rho$. In preparation for this, we define the state $\bmstate{k}$ and probability $\bmaccept{k}$ to be the blended measurement analogues of the state $\state{k}$ and probability $\accept{k}$ introduced in \Cref{sec:Introduction}.

\begin{defn} \label{def:blend_accept_and_state}
Given a state $\rho$ and set of two outcome measurements $\cM$ let the state $\bmstate{k}$ be the resulting state when the measurement $\blend(\cM)$ is applied $k$ times in sequence to a quantum system initially in state $\rho$ and the reject outcome is observed each time, so
\begin{align}
    \bmstate{k} = \frac{E_0^{k} \rho E_0^{k}}{\Tr[E_0^{k} \rho E_0^{k}]}\,.
\end{align}
Let $\bmaccept{k}$ be the probability that at least one accepting outcome is observed when the measurement $\blend(\cM)$ is applied $k$ times in sequence to a quantum system in state $\rho$ (equivalently, the probability that not all outcomes observed are reject), so
\begin{align}
    \bmaccept{k} = 1 - \Tr[E_0^{k} \rho E_0^{k}]\,. \label{eq:bacept_base_defn}
\end{align}
When the set of measurements $\cM$ is clear from context we will refer to these objects using the simplified notation $\bstate{k}$ and $\baccept{k}$. 
\end{defn}

\begin{rmk} \label{rmk:baccept_rewrite}
We can also write $\baccept{k}$ as 
\begin{align}
    \baccept{k} &= \sum_{i=0}^{k-1} \Tr[(1-E_0^2) E_0^{i} \rho E_0^{i}] \\
    &= \sum_{i=0}^{k-1} \left(1 - \baccept{i} \right)\Tr[(1-E_0^2) \bstate{i}]\,.
\end{align}
Written this way it is clear that $\baccept{k}$ is equal to the sum from $i = 0$ to $k-1$ of the probability that a repeated blended measurement accepts for the first time at round $i$.  
\end{rmk}

We note that (unlike in the random measurements case) the states $\rho$ and $\bstate{k}$ are related in a very simple way -- via conditioning on the single PSD matrix $E_0^{k}$. We can use this observation to prove some basic results about $\bstate{k}$ and $\baccept{k}$. 

\begin{lem}[Gentle Blended Measurements] \label{thm:gentle_blended}
Let $\rho$ be a state, $\cM$ be a set of two outcome measurements and define $\bstate{k}$ and $\baccept{k}$ as in \Cref{def:blend_accept_and_state}. Then,
\begin{align}
    \norm{\bstate{k} - \rho}_1 \leq 2 \sqrt{\baccept{k}}\,.
\end{align}
\end{lem}
\begin{proof}
Immediate from the gentle measurement lemma. 
\end{proof}

\begin{lem}
\label{lem:blended_nondecreasing}
For any blended measurement $\cB(\cM)$ and states $\bstate{k}$ defined as in \Cref{thm:gentle_blended} we have 
\begin{align}
    \Tr[E_0^2 \bstate{k}] \geq \Tr[E_0^2 \bstate{k - 1}]
\end{align}
where $k \geq 1$ and $\bstate{0} = \rho$. 
\end{lem}

\begin{proof}
Immediate since conditioning on a measurement outcome can only increase the probability of that measurement outcome occurring again. Details are in \Cref{cor:many_repeated_measurements_increase_prob} in the appendix. 
\end{proof}

One consequence of \Cref{lem:blended_nondecreasing} is the following simple upper bound on $\baccept{k}$.

\begin{cor}
\label{cor:blended_accept_bound}
Define $\baccept{k}$ as in \Cref{thm:gentle_blended} and let 
\begin{align}
    \epsilon = \frac{1}{m} \Tr[\sum_i M_i \rho]\,. 
\end{align}
Then 
\begin{align}
    \baccept{k} \leq k \epsilon\,.
\end{align}
\end{cor}

\begin{proof}
Writing $\baccept{k}$ as a sum of accepting probabilities at each round (see \Cref{rmk:baccept_rewrite}) we have
\begin{align}
    \baccept{k} &= \sum_{i = 0}^{k-1}(1-\baccept{i})\Tr[(1-E_0^2)\bstate{i}] \\
    &\leq \sum_{i = 0}^{k-1}\Tr[(1-E_0^2)\bstate{i}] \\
                &\leq k \Tr[(1-E_0^2)\bstate{0}] = k \Tr[\sum_{i=1}^{m} E_i^2 \rho] = k \epsilon\,,
\end{align}
where we used \Cref{lem:blended_nondecreasing} on the 3rd line, as well as the definitions of $E_0$ and $E_i$. 
\end{proof}

\subsection{Gentle Random Measurements}
\label{sec:random_measurements}

In this section we apply the tools from \Cref{sec:technicals}, and \Cref{sec:blended_measurements} to attain bounds on the distrubance caused by and accepting probability of repeated random measurements. First, we show how the notation introduced in \Cref{defn:multiMeasurement} can be used to describe the quantities $\accept{k}$ and $\state{k}$ defined in \Cref{sec:Introduction} as well as the quantity $\baccept{k}$ defined in \Cref{sec:blended_measurements}. 

\begin{rmk} \label{claim:accept(k)state(k)} \label{obs:random_measurements}
Let $\cM = \{M_1, M_2, \ldots , M_m\}$ be a set of two outcome projective measurements and $\rho$ be a state. Consider the process where measurements are drawn at random from $\cM$ and applied in sequence to a quantum system originally in state $\rho$. Recall from \Cref{sec:Introduction} that $\accept{k}$ gives the probability that at least one of $k$ random measurements applied to the system accepts, and $\state{k}$ gives the state of the quantum system conditioned on the event that $k$ random measurements are applied to the system and none accept. Define the set of matrices $\notcM = \{1-M_1, 1-M_2, \ldots , 1-M_m\}$. Then 
\begin{align}
    \accept{k} &= 1 - \expec_{M_1, \ldots, M_k \sim \cM} \left[\Tr[ (1-M_k) \ldots (1 - M_1) \rho (1 - M_1) \ldots (1-M_k)]\right]\\
    &= 1 - \frac{1}{m^k} \sum_{T \in \mm{\notcM}{k}} \Tr[ T \rho T\adj ] \label{eq:Accept(k)}
\end{align}
and 
\begin{align}
\rho^{(k)} &= \frac{\expec_{M_1, \ldots, M_k \sim \cM} \left[ (1-M_k) \ldots (1 - M_1) \rho (1 - M_1) \ldots (1-M_k) \right]}{\expec_{M_1, \ldots, M_k \sim \cM} \left[\Tr\left[ (1-M_k) \ldots (1 - M_1) \rho (1 - M_1) \ldots (1-M_k) \right]\right] } \\
&= \frac{1}{m^k (1 - \accept{k})} \sum_{T \in \mm{\notcM}{k}} T \rho T\adj\,. \label{eq:State(k)}
\end{align}
Similarly, we have 
\begin{align}
    \baccept{k} &= 1 - \Tr[E_0^{2k} \rho] \\
    &= 1 - \frac{1}{m^k} \sum_{T \in \mm{\notcM}{k}} \Tr[T \rho]\,.
\end{align}

\end{rmk}

Now we being the task of relating the random and blended measurement procedures. Our first bound shows that when $k$ sequential blended measurements are unlikely to accept on a state $\rho$ (that is, $\baccept{k}$ is small) then the damage caused to the state by $k$ sequential random measurements is also bounded. 

\begin{lem}
\label{lem:random_fidelity_bound}
Given a state $\rho$ and set of two outcome projective measurements $\cM$, define $\state{k}$ as in \Cref{sec:Introduction} (so $\state{k}$ gives the state of the system initially in state $\rho$ after $k$ random measurements reject) and $\baccept{k}$ as in \Cref{sec:blended_measurements} (so $\baccept{k}$ gives the probability that at least one of $k$ repeated blended measurements applied to $\rho$ accepts). Then 
\begin{align}
    \fidelity{\state{k}}{\rho} \geq 1 - \baccept{k}
\end{align}
\end{lem}

\begin{proof}
We first prove the result when $\rho = \ketbra{\psi}$ is a pure state. In that case, we find
\begin{align}
    \fidelity{\rho} {\rho^{(k)}}^2 &= \ev{\rho^{(k)}}{\psi} \\
    &= \frac{1}{m^k} \sum_{T \in \mm{\notcM}{k}} \Tr[\ketbra{\psi} T \ketbra{\psi} T\adj] (1 - \accept{k})^{-1} \\    
    &\geq \frac{1}{m^{2k}} \sum_{T, S \in \mm{\notcM}{k}}\Tr[\ketbra{\psi} T \ketbra{\psi} S ] (1 - \accept{k})^{-1} \\
    &= \frac{1}{m^{2k}} \sum_{T, S \in \mm{\notcM}{k}}\Tr[ T \ketbra{\psi} ]\Tr[S \ketbra{\psi}] (1 - \accept{k})^{-1} \\
    &= \left(\frac{1}{m^{k}} \sum_{T \in \mm{\notcM}{k}}\Tr[ T \ketbra{\psi} ]\right)^2 (1 - \accept{k})^{-1} \\
    &= \left(1 - \baccept{k}\right)^2\left(1 - \accept{k}\right)^{-1} \geq \left(1 - \baccept{k} \right)^2\,.
\end{align}
In the derivation above we used \Cref{obs:random_measurements} (\Cref{eq:State(k)}) to go from the first line to the second, \Cref{cor:taking_avg_outside_trace_specific} to go from the second line to the third.  Then we re-organize terms and replace the definition of $\baccept{k}$ from \Cref{def:blend_accept_and_state} (\Cref{eq:bacept_base_defn}) to get the desired result.

Now (as in the proof of the Gentle Measurement Lemma), when $\rho$ is a mixed state we can recover the same bound by first purifying $\rho$. Formally, let $\system{A}$ be the quantum system that $\rho$ and the measurements $\cM$ act on. Then there exists another quantum system $\system{R}$ and a pure state $\ket{\psi}_{\system{AR}}$ such that
\begin{equation}
    \Tr_{\system{R}}[\ketbra{\psi}_{\system{AR}}] = \rho\,.
\end{equation}
Let $\tilde{\rho} = \ketbra{\psi}_{\system{AR}}$ and \begin{align}
    \tilde{\rho}^{(k)} = \frac{1}{m^k}\sum_{T \in \mm{\overline{M}}{k}} \left(T_{\system{A}} \otimes I_{\system{R}}\right) \tilde{\rho} \left(T_{\system{A}} \otimes I_{\system{R}}\right)\adj (1 - \accept{k})^{-1}\,.
\end{align}
Note that by definition we have that 
$\Tr_\system{R}[\tilde{\rho}^{(k)}] = \rho^{(k)}.$ Then we have 
\begin{align}
    \fidelity{\rho}{ \state{k}}^2 \geq \fidelity{\tilde{\rho}}{ \tilde{\rho}^{(k)}}^2 
\end{align}
since tracing out a subsystem can only increase the fidelity between two states, and 
\begin{align}
    \fidelity{\tilde{\rho}}{ \tilde{\rho}^{(k)}}^2 
    &= \frac{1}{m^k}\left(\sum_{T \in \mm{\notcM}{k}}\Tr[\ketbra{\psi}_{\system{AR}} (T_{\system{A}} \otimes I_{\system{R}}) \ketbra{\psi}_{\system{AR}} (T_{\system{A}} \otimes I_{\system{R}})^\dagger]\right)(1 - \accept{k})^{-1}
    \\&\geq  \left(m^{-k} \sum_{T \in \mm{\notcM}{k}}\Tr[ \left(T_\system{A} \otimes I_\system{R}\right) \ketbra{\psi}_{\system{AR}} ]\right)^2 (1 - \accept{k})^{-1} \\
    &= \left(m^{-k} \sum_{T \in \mm{\notcM}{k}}\Tr[ T \rho]\right)^2 (1 - \accept{k})^{-1} \\
    &\geq (1 - \baccept{k})^2
\end{align}
by the argument above. Combining these two bounds completes the proof. 
\end{proof}

With this we can bound the damage caused by random measurements by the accepting probability of a blended measurement procedure.  
\begin{cor}
\label{cor:random_blended_trace_bound}
Given a state $\rho$ and set of two outcome projective measurements $\cM$, define $\state{k}$ as in \Cref{sec:Introduction} and $\baccept{k}$ as in \Cref{sec:blended_measurements}. Then
\begin{equation}
    \norm{\rho - \state{k}}_1 \leq 2 \sqrt{2\baccept{k}}\,.
\end{equation}
\end{cor}
\begin{proof}
The standard relationship between trace distance and fidelity, along with \Cref{lem:random_fidelity_bound} gives 
\begin{align}
    \norm{\rho - \state{k}}_1 \leq 2\sqrt{1 - \fidelity{\rho}{ \state{k}}^2} \leq 2 \sqrt{2 \baccept{k}}\,,
\end{align}
as desired. 
\end{proof}

Now we relate the acceptance probability of the random measurement procedure to the accept probability of the blended measurement procedure. We begin with a slight restatement of \Cref{cor:taking_avg_outside_trace_specific} which gives a relationship between the probability of measurement outcomes being observed on states $\state{k}$ and $\bstate{2k}$.

\begin{rmk}\label{rmk:long_form_rho}
We can expand out the definition of $\bstate{2k}$ in \Cref{def:blend_accept_and_state} to get the following formulation:
\begin{align}
    \bstate{2k} &= \frac{E_0^{2k} \rho E_0^{2k}}{1 - \baccept{2k}} \\
    &= \frac{m^{-2k}}{1 - \baccept{2k}}  \left(\sum_{i = 0}^m 1 - M_i\right)^k \rho \left(\sum_{j = 0}^m 1 - M_j\right)^k\\
    &= \frac{m^{-2k}}{1 - \baccept{2k}} \left(\sum_{T, S \in \mm{\notcM}{k}} T \rho S\right)\,.
\end{align}
\end{rmk}

\begin{cor} \label{cor:bounded_measurement_outcomes}
For any state $\rho$ and set of two outcome projective measurements $\cM$ define states $\state{k}$, $\bstate{k}$ and probabilities $\accept{k}$ and $\baccept{k}$ as in \Cref{sec:Introduction} and \Cref{sec:blended_measurements}. Also, let $X$ be an arbitrary PSD matrix. Then 
\begin{align}
    (1 - \accept{k}) \Tr[X \state{k}] \geq (1 - \baccept{2k}) \Tr[X \bstate{2k}]\,.
\end{align}
\end{cor}

\begin{proof}
We compute 
\begin{align}
    (1 - \accept{k}) \Tr[X \state{k}] &= \frac{1}{m^k} \sum_{T \in \mm{\notcM}{k}} \Tr[X T \rho T\adj] \\
    &\geq m^{-2k} \sum_{T, S \in \mm{\notcM}{k}} \Tr[X T \rho S] \\
    &= (1 - \baccept{2k})\Tr[X \bstate{2k}]
\end{align}
where we used \Cref{obs:random_measurements} (\Cref{eq:State(k)}) on the first line, \Cref{cor:taking_avg_outside_trace_specific} on the second line, and \Cref{rmk:long_form_rho} on the third.

\end{proof}

Next, we show that \Cref{cor:bounded_measurement_outcomes} gives an easy upper bound $\accept{k}$ in terms of $\baccept{k}$. This bound is not required for the proof of \Cref{thm:random_gentle_total_accept}, but does give a useful relationship between the random and blended measurement procedures which we will use in future sections.

\begin{thm} \label{thm:accept_upper_bound}
For any state $\rho$ and set of two outcome projective measurements $\cM$ define $\accept{k}$, $\baccept{k}$ as in \Cref{sec:Introduction} and \Cref{sec:blended_measurements}.  Then we have
\begin{align}
    1 - \accept{k} \geq 1 - \baccept{2k} \geq (1 - \baccept{k})^2\,.
\end{align}
\end{thm}
\begin{proof}
Using \Cref{cor:bounded_measurement_outcomes} with $X = I$ and noting that $\state{k}$, $\bstate{2k}$ are both normalized density matrices gives
\begin{align}
    1 - \accept{k} &= (1 - \accept{k}) \Tr[\state{k}] \\
    &\geq (1 - \baccept{2k}) \Tr[\bstate{2k}] = (1 - \baccept{2k}) 
\end{align}
which completes the proof of the first inequality. 

To prove the second inequality, note that
\begin{align}
    (1 - \baccept{k}) &= \Tr[(E_0)^{2k} \rho] \\
    &\leq  \sqrt{\Tr[(E_0)^{4k} \rho] \Tr[\rho]}\\
    &= \sqrt{1 - \baccept{2k}}\,.
\end{align}
Where we used Cauchy-Schwarz on the second line and the definition of $\baccept{k}$ (\Cref{eq:bacept_base_defn}) on the first and third lines. 

\end{proof}
 We can also use \Cref{cor:bounded_measurement_outcomes} to lower bound of $\accept{k}$ in terms of $\baccept{k}$. This is the direction required for the proof of \Cref{thm:random_gentle_total_accept}.

\begin{thm} \label{thm:random_measurements_upper_bound}
For any state $\rho$ and set of two outcome projective measurements $\cM$ define $\accept{k}$, $\baccept{k}$ as in \Cref{sec:Introduction} and \Cref{sec:blended_measurements}. We have
\begin{align}
    \accept{k} \geq \frac{1}{2} \baccept{2k}
\end{align}
\end{thm}

\begin{proof}
Then, define
\begin{align}
    \avM = \sum_i M_i/m = \sum_i E_i^2 = 1 - E_0^2\,.
\end{align}
Note that for any state $\sigma$, $\Tr[\avM \sigma]$ gives the probability that the blended measurement $\blend(\cM)$ results in an accepting outcome, which is equal to the probability that a measurement chosen uniformly at random from $\cM$ accepts on $\sigma$. 

Then we calculate
\begin{align}
    \accept{k} 
    &= \sum_{i = 0}^{k-1} (1 - \accept{i}) \Tr(\avM \state{i}) \\
    &\geq \sum_{i = 0}^{k-1} (1 - \baccept{2i}) \Tr(\avM \bstate{2i}) \\
    &\geq \frac{1}{2} \sum_{j = 0}^{2k-1} (1 - \baccept{j}) \Tr(\avM \bstate{j}) \\
    &= \frac{1}{2} \baccept{2k}\,.
\end{align}
The first and last lines follows from a telescoping sums argument for both blended and randomized measurements.  The second line is a direct application of \Cref{cor:bounded_measurement_outcomes}, The final line follows from \Cref{lem:blended_nondecreasing}, which shows that $\Tr[\avM \bstate{j}] = \Tr[(1 - E_0) \bstate{j}]$ is a decreasing function of $j$, plus the observation that $\baccept{j}$ is an increasing function of $j$ by definition.  This implies that $(1 - \baccept{j})\Tr[\avM \bstate{j}] + (1 - \baccept{j+1})\Tr[\avM \bstate{j+1}] \leq 2 (1 - \baccept{j})\Tr[\avM \bstate{j}]$, so we can fill in the odd indexed terms in the sum. 
\end{proof}

\begin{rmk} \label{note:accept_vs_baccept}
Putting together \Cref{thm:accept_upper_bound} and \Cref{thm:random_measurements_upper_bound} gives 
\begin{align}
   \frac{1}{2}\baccept{2k} \leq \accept{k} \leq \baccept{2k} 
\end{align}
which gives a reasonably tight bound on the probability of $k$ random measurements accepting.
\end{rmk}

Finally, we are in a position to prove \Cref{thm:random_gentle_total_accept}. We begin by repeating the theorem.

\begin{reptheorem}{thm:random_gentle_total_accept}
Let $\cM = \{M_1, M_2, \ldots, M_m\}$ be a set of two outcome projective measurements, and $\rho$ be a pure state. Consider the process where a measurement from the set $\cM$ is selected universally at random and applied to a quantum system initially in state $\rho$. Let $\accept{k}$ be the probability that at least one measurement accepts after $k$ repetitions of this process, and let $\rho^{(k)}$ be the state of this quantum system after $k$ repetitions where no measurement accepts. Then
\begin{align}
\norm{\rho - \state{k}}_1 \leq 4 \sqrt{\accept{\lceil k/2 \rceil}} \leq 4 \sqrt{\accept{k}}\,.
\end{align}
\end{reptheorem}

\begin{proof}
Combining \Cref{cor:random_blended_trace_bound} and \Cref{thm:random_measurements_upper_bound} gives
\begin{align}
    \norm{\rho - \state{k}}_1 \leq 2\sqrt{2\baccept{k}} \leq 4\sqrt{\accept{\lceil k/2 \rceil } }  \leq 4\sqrt{\accept{k}}\,.
\end{align}
\end{proof}

\begin{rmk} It might seem lossy to go directly from $\sqrt{\accept{\lceil k/2 \rceil } }$ to $\sqrt{\accept{k}}$ in the above bound, but we can't do better in general since there are cases (for example, the set $\cM$ contains a single projector) where once a random measurement has rejected once it will reject forever and so $\accept{k} = \accept{\lceil k/2 \rceil} =  \accept{1}$ and the inequality is tight. 
\end{rmk} 

\section{Algorithms for Quantum OR}
\label{sec:Quantum_OR}

In our first application of the results from \Cref{sec:blended_measurements} and \Cref{sec:random_measurements}, we give two different procedures for Quantum OR.  We call a procedure a ``Quantum OR'' if it has properties similar to Corollary 11 from Ref.~\cite{harrow2017sequential}, which we restate here.
\begin{thm}[Corollary 11 From Ref.~\cite{harrow2017sequential}] \label{thm:oldOR} Let $\Lambda_1, \Lambda_2, \ldots , \Lambda_m$ be a sequence of projectors and fix $\epsilon > 1/2, \delta$. Let $\rho$ be a state such that either there exists an $i \in [m]$ with $\Tr[\Lambda_i \rho] > 1- \epsilon$ (Case 1) or $\mathbb{E}_j[\Tr[\Lambda_j \rho]] \leq \delta$ (Case 2). Then there exists a test that uses one copy of $\rho$ and: in Case 1, accepts with probability $(1-\epsilon)^2/7$; in Case 2, accepts with probability at most $4\delta m$.
\end{thm}

\subsection{Repeated Blended Measurements}

\label{subsec:Blend_Measurements_OR}
We first show that repeated application of the blended measurement defined in \Cref{sec:blended_measurements} yields a Quantum OR protocol.  We define the protocol next.  

\begin{framed}
\begin{algorithm}\label{algorithm:blended_quantum_or}
\textbf{Blended Measurement Quantum OR}
\end{algorithm}
\noindent \textbf{Input:} A classical description of a set of two outcome measurements $\cM = \{M_1, M_2, \ldots , M_m\}$ and a single copy of a state $\rho$. \mbox{} \vspace{3mm} \newline
\textbf{Output:} \ACCEPT or \REJECT. 
\begin{enumerate}
    \item Prepare a quantum system in state $\rho$.
    \item Repeat $m$ times:
    \begin{enumerate}
        \item Perform the blended measurement $\blend(\cM)$ on the state. If the measurement accepts, return \ACCEPT. 
    \end{enumerate}
    \item Return \REJECT.
\end{enumerate}
\end{framed}

The following result shows that \cref{algorithm:blended_quantum_or} solves the Quantum OR problem and obtains better parameters than the protocol given in Ref.~\cite{harrow2017sequential}.

\begin{thm}[Blended Quantum OR] \label{thm:Blended_Quantum_OR}
Let $\cM = \{M_1, M_2, \ldots , M_m\}$ be a set of two outcome measurements and let $\rho$ be a quantum state. Define 
\begin{align}
    \plob &= \max_i\left\{\Tr[M_i \rho] \right\}\,,\\
    \pupb &= \sum_i \Tr[M_i \rho]\,,\text{ and } \\
    \paccept &= \mathbb{P}\left(\text{\cref{algorithm:blended_quantum_or} accepts with input } \cM \text{ and } \rho\right)\,. 
\end{align}
Then the following inequalities hold:
\begin{align}
  \plob^2/4 < \paccept < \pupb\,.
\end{align}
\end{thm}

\begin{proof}
We first prove the upper bound. Let $\preject = 1 - \paccept$ be the probability the algorithm rejects, and let $\preject(k)$ be the probability that the algorithm does not accept on the $k^{th}$ measurement, conditioned on the algorithm not accepting any of the $i-1$ measurements prior. We note 
\begin{align}
    \preject(1) &= \frac{1}{m}\sum_i\left(1 - \Tr[M_i \rho]\right) = 1 - \frac{\pupb}{m}
\end{align}
and, by \Cref{cor:many_repeated_measurements_increase_prob}, $\preject(k) \geq \preject(1)$. Then
\begin{align}
    \paccept &= 1 - \preject \\
    &= 1 - \prod_{k=1}^m \preject(k) \\
    &\leq 1 - \left(1 - \frac{\pupb}{m} \right)^m  \\
    &\leq 1 - e^{-\pupb} \leq \pupb\,.
\end{align}
The second line is a result of \Cref{lem:blended_nondecreasing}. The third inequality follows from the definition of $e^x$ and the final inequality follows from the inequality $1+x \leq e^x$. 

We now prove the lower bound. First, for ease of notation, we relabel the measurements in $\cM$ so that $\plob = \Tr[M_1 \rho]$. Then 
let $\breject{k}$ be the probability that the first $k$ measurements of \cref{algorithm:blended_quantum_or} reject (with $\breject{0} = 1$), and $\rhobreject{k}$ be the state of the quantum system initially in state $\rho$ conditioned on the first $k$ measurements of \cref{algorithm:blended_quantum_or} rejecting (with $\rhobreject{0} = \rho$). By definition, the probability of accepting is at least the sum over all rounds of the probability the algorithm accepts for the first time on a given round with measurement $M_1$, so 
\begin{align}
    \paccept &\geq \frac{1}{m} \sum_{k=0}^{m-1} \breject{k} \Tr[M_1 \rhobreject{k}] \label{eq:paccept_lower_bound}\,.
\end{align}
In order to return reject the algorithm must at least reject on the first $k$ measurements, so
\begin{equation}
    \breject{k} \geq \preject = 1 - \paccept\label{eq:Ri_lower_bound}\,.
\end{equation}
The probability of accepting on any of the first $k$ measurements is $1 - \breject{k}$, and it is known that for any two states $\rho$ and $\sigma$ and PSD matrix $M \leq 1$ 
\begin{align}
    \left|\Tr\left[M(\rho - \sigma)\right]\right| \leq \frac{1}{2}\norm{\rho - \sigma}_1\,. 
\end{align}  
Then applying  \Cref{thm:gentle_blended} (The Gentle Blended Measurement Lemma) to \Cref{eq:paccept_lower_bound} along with the previous two facts gives  
\begin{align}
    \paccept &\geq \frac{1}{m} \sum_{k=0}^{m-1} \breject{k} \left(\Tr[M_1 \rho] - \sqrt{1-\breject{k}} \right) \\
    &\geq \frac{1}{m} \sum_{k=0}^{m-1} (1-\paccept)\left(\plob - \sqrt{1 - \breject{k}} \right)\\
    &\geq (1-\paccept)\left(\plob - \sqrt{\paccept} \right)\,,
\end{align}
where the last two lines simply apply \Cref{eq:Ri_lower_bound}. 
 Rearranging terms we arrive at the following 
\begin{equation}
\plob \leq \frac{\paccept}{1-\paccept} + \sqrt{\paccept}\,.
\end{equation}

We note that $\frac{x}{1-x} \leq \sqrt{x}$ whenever $x \leq \frac{1}{2}(3-\sqrt{5}) \approx 0.38$ (take $x$ to be $\paccept$).  Therefore, if $\paccept \leq 0.38$, we have 
\begin{align}
    \plob&\leq 2\sqrt{\paccept}\\
    \implies \paccept &\geq \plob^2/4\,.
\end{align}
If $\paccept$ is greater than $0.38$ then it is still larger than $\min(\plob^2/4, 0.38) = \plob^2/4$, which completes the lower bound. 
\end{proof}

\begin{cor}
In the same setting as \Cref{thm:oldOR}, but with $\Lambda_1, \ldots, \Lambda_m$ arbitrary (i.e. not necessarily projective) two outcome measurements, there exists a test that uses one copy of $\rho$ and accepts with probability at least $(1 - \epsilon)^2/4$ in case $1$ and at most $\delta n$ in case $2$.  
\end{cor}
\begin{proof}
\Cref{thm:Blended_Quantum_OR} shows that \cref{algorithm:blended_quantum_or} satisfies the required bounds in cases $1$ and $2$.  
\end{proof}

\subsection{Repeated Random Measurements}
Motivated by the original Quantum OR claimed in Ref.~\cite{aaronson2006qma}, we show that repeated random measurements still yield a (weaker) Quantum OR.  The Random Measurement Quantum OR protocol is described in \cref{algorithm:random_quantum_or}.

\begin{framed}
\begin{algorithm}\label{algorithm:random_quantum_or}
    \textbf{Random Measurement Quantum OR}
\end{algorithm}

\noindent \textbf{Input:} A black-box implementation of each measurement in a set of two outcome measurements $\cM = \{M_1, M_2, \ldots , M_m\}$ and a single copy of a state $\rho$. \mbox{} \vspace{3mm} \newline
\textbf{Output:} \ACCEPT or \REJECT. 
\begin{enumerate}
    \item Prepare a quantum system in state $\rho$.
    \item Repeat $m$ times:
    \begin{enumerate}
        \item Pick a random measurement $M_i \in \cM$.
        \item Perform the measurement $M_i$ on the current state. If the measurement accepts, return \ACCEPT. 
    \end{enumerate}
    \item Return \REJECT.
\end{enumerate}
\end{framed}

\begin{thm}[Random Quantum OR]\label{thm:Random_Quantum_OR}
Let $\cM = \{M_1, M_2, \ldots, M_m\}$ be a set of two outcome projective measurements, and $\rho$ be a state. Then using the same definitions for $\plob$ and $\pupb$ as in \Cref{thm:Blended_Quantum_OR} and letting $\paccept = \mathbb{P}(\text{\cref{algorithm:random_quantum_or} accepts with inputs } \cM \text{ and } \rho)$, the following inequalities hold:
\begin{equation}
    \min\left(\plob^2/4.5, \frac{3-\sqrt{5}}{4}\right) \leq \paccept \leq 2\pupb\,.
\end{equation}
\end{thm}
\begin{proof}
Similarly to last time, we first prove the upper bound.  By \Cref{thm:accept_upper_bound}, $\paccept$ is upper bounded by the probability the blended measurement $\cB(\cM)$ accepts at least once after being applied $2m$ many times.  From the same argument as \Cref{thm:Blended_Quantum_OR}, we have that
\begin{align}
    \paccept &\leq 1 - \left(1 - \frac{\pupb}{m}\right)^{2m}\\
    &\leq 1 - e^{-2\pupb} \\&\leq 2\pupb\,.
\end{align}

For the lower bound we will use the same reasoning as the blended case, after relating the random and blended measurement procedures using \Cref{note:accept_vs_baccept}.  In particular, we are interested in finding the probability that performing $2k$ blended measurements accepts.  As in \Cref{subsec:Blend_Measurements_OR}, define $\breject{k}$ to be the probability that all of the first $k$ (inclusive) blended measurements reject, and define $\pbaccept$ to be the probability of $2m$ blended measurements accepting.  By the same reasoning as before we have that
\begin{align}
    \pbaccept &\geq \frac{1}{m} \sum_{k = 0}^{2m-1} \breject{k} \left(\Tr[M_1 \rho] - \sqrt{1 - \breject{k}}\right)\, \text{ and}\\ \breject{k} &\geq 1 - \pbaccept\,. \label{eq:random_measurement_inequalities}
\end{align}
Putting these together gives
\begin{align}
    \pbaccept &\geq \sum_{k=0}^{2m-1}\frac{\breject{k}}{m}\left(\Tr(M_1 \rho) - \sqrt{1-\breject{k}}\right)\\
    &\geq 2(1-\pbaccept)\left(\plob - \sqrt{\pbaccept}\right)\,.
\end{align}
Rearranging terms, we get the following inequality
\begin{equation}
    \plob \leq \frac{\pbaccept}{2(1 - \pbaccept)} + \sqrt{\pbaccept}\,.
\end{equation}
Thus, when $\pbaccept \geq \frac{3-\sqrt{5}}{2} \approx 0.38$, we have that 
\begin{equation}
    \pbaccept \geq \plob^2/1.5^2\,.
\end{equation}
Plugging this into \Cref{note:accept_vs_baccept}, we find that when $\pbaccept \geq 0.38$
\begin{equation}
    \paccept \geq \frac{1}{2}\pbaccept = \frac{\plob^2}{4.5}\,.
\end{equation}
Thus, $\paccept \geq \min(\plob^2/4.5, 0.19)$.  Finally note that $\plob^2/4.5 \leq 0.19$, so we have that $\paccept \geq \plob^2/4.5$, which completes the proof of the lower bound. 
\end{proof}

We note that the Random Quantum OR performs worse than the Blended Quantum OR on both the accept and reject case, but performs better in both cases than the test from Ref.~\cite{harrow2017sequential}.  The Random Quantum OR has additional advantages over both protocols, in that it does not require knowledge of a circuit decomposition of the measurements $M_i$ and can even apply the measurements $M_i$ as a black box.  

\begin{cor}
In the same setting as \Cref{thm:oldOR}, there exists a test that uses one copy of $\rho$, does not require an efficient representation of the measurements $\Lambda_i$, and accepts with probability at least $\frac{3-\sqrt{5}}{4} (1 - \epsilon)^2$ in case $1$ and at most $2\delta n$ in case $2$.  
\end{cor}
\begin{proof}
\Cref{thm:Random_Quantum_OR} shows that \cref{algorithm:random_quantum_or} satisfies the required bounds in cases $1$ and $2$, since $\frac{3-\sqrt{5}}{4}\plob^2 \leq \min\left(\plob^2/4.5, \frac{3-\sqrt{5}}{4}\right)$.  
\end{proof}

\section{Quantum Event Finding}
\label{sec:quantum_event_finding}

In this section we consider a variant of the quantum OR task in which the goal, given a set of two outcome measurements $\cM$ and sample access to a state $\rho$, is not just to \textit{decide} if there exists a measurement $M_i \in \cM$ with $\Tr[M_i \rho]$ large, but also to \textit{find} such a measurement if one exists. We show that the blended measurement procedure described in the previous section solves with problem if certain conditions are met. 

\begin{thm}
\label{thm:blended_event_finding}
Let $\cM = \{M_1, M_2, \ldots , M_m\}$ be a set of two outcome measurements. Let $\rho$ be a state such that either there exists an $i \in [m]$ with $\Tr[M_i \rho] > 1- \epsilon$ (Case 1) or $\sum_i \Tr[M_i \rho] \leq \delta$ (Case 2). Also define 
\begin{align}
    \smallWeight &= \sum_{i :\; \Tr[M_i \rho] < 1 - \epsilon} \Tr[M_i \rho]\,.
\end{align}
Then if the blended measurement $\blend(\cM)$ is applied $m$ times in sequence to a quantum system initially in state $\rho$: in Case 1, with probability at least 
\begin{align}
    \frac{(1 -\epsilon)^3}{12 (1 + \beta)}\,,
\end{align} 
at least one accepting outcome is observed and the first accepting outcome observed corresponds to a measurement $M_i$ with $\Tr[M_i \rho] > 1-\epsilon$; in Case 2, an accepting outcome is observed with probability at most $\delta$.
\end{thm}

By relating blended and random measurements as in the proof of \Cref{thm:Random_Quantum_OR}, we can also show that the random measurement procedure solves this problem in the same regime as before, but with worse constants and scaling in both $\epsilon$ and $\beta$.  

\begin{thm}
\label{thm:random_event_finding}
Let $\cM = \{M_1, M_2, \ldots, M_m\}$ be a set of two outcome projective measurements, and define $\rho$, $\smallWeight$, $\epsilon$ and $\delta$ as above. Then, if measurements are chosen at random (with replacement) from $\cM$ and applied to a quantum system initially in state $\rho$: in Case 1, with probability at least $(1 -\epsilon)^{7}/(1296 (1 + \beta)^3)$, at least one measurement accepts and the first accepting measurement is a measurement $M_i \in \cM$ with $\Tr[M_i \rho] > 1- \epsilon$; in Case 2, a measurement accepts with probability at most $2\delta$. 
\end{thm}

We begin by proving \Cref{thm:blended_event_finding}.

\begin{proof}[Proof (\Cref{thm:blended_event_finding})]

The upper bound on the accepting probability in Case 2 follows immediately from the upper bound on the accepting probability in Case 2 of the blended measurement quantum OR procedure stated in \Cref{thm:Blended_Quantum_OR}. 

To prove the lower bound in Case 1 we follow a procedure similar to the one used in the proof of \Cref{thm:Blended_Quantum_OR}. First, for ease of notation, relabel measurements so that 
\begin{align}
    \Tr[M_j \rho] > 1- \epsilon
\end{align}
if and only if $j \leq k$ for some constant $k$. Similar to before, let $\breject{i}$ be the probability that the blended measurement $\blend(\cM)$ is applied $i$ times in sequence to a quantum system initially in state $\rho$ and no measurement accepts, let $\baccept{i} = 1 - \breject{i}$, and let $\rhobreject{i}$ be the state of the quantum system after $i$ blended measurements all reject. Also let $\firstAcceptB{i}$ be the \textbf{event} that the blended measurement procedure accepts for the first time on the $i$'th measurement and $\goodFirstAcceptB{i}$ be the \textbf{event} that the blended measurement procedure accepts for the first time on the $i$th measurement on a outcome corresponding to a measurement $M_j$ with $\Tr[M_j \rho] > 1-\epsilon$. Then we can lower bound the probability of success on a measurement conditioned on no previous measurement accepting
\begin{align}
    \p\left[\goodFirstAcceptB{i} \; \Big| \bigwedge_{j < i} \neg \firstAcceptB{j} \right] &= \sum_{j \leq k} \frac{1}{m} \Tr[M_j \rhobreject{i-1}] \\&\geq \frac{k}{m} \left(1 - \epsilon - \sqrt{\baccept{i-1}}\right)
\end{align}
where the inequality follows from \Cref{thm:gentle_blended}.
Additionally, \Cref{lem:blended_nondecreasing} tells us that 
\begin{align}
    \p[\firstAcceptB{i}] &\leq \p[\firstAcceptB{0}] \\
    &= 
    \frac{1}{m} \sum_i \Tr[M_i \rho]\\ 
    &\leq (k + \beta)/m\,.
\end{align}
Combining these two bounds gives
\begin{align}
    \p[\goodFirstAcceptB{i}|\firstAcceptB{i}] &\geq \frac{k (1 - \epsilon - \sqrt{\baccept{i-1}})}{k + \beta} \\
    &\geq \frac{1 - \epsilon - \sqrt{\baccept{i-1}}}{1 + \beta}\,.\label{eq:QEF_one_round_good_accept_probability}
\end{align}
But we also have that 
\begin{equation}
    \p[\firstAcceptB{i}] = \baccept{i} - \baccept{i-1}\,.
\end{equation}
Then we can bound the overall fraction of the first accepting events in which the accepting outcome corresponds to a measurement $M_i$ with $\Tr[M_i \rho] > 1-\epsilon$ as 
\begin{align}
    &\frac{\sum_{i=1}^m \p[\goodFirstAcceptB{i}]}{\sum_{i=1}^{m} \p[\firstAcceptB{i}]} \\
    &\hspace{5pt}= \frac{\sum_{i=1}^{m} \p[\firstAcceptB{i}]\p[\goodFirstAcceptB{i}|\firstAcceptB{i}]}{\sum_{i=1}^{m} \p[\firstAcceptB{i}]}\\
    &\hspace{5pt}= \frac{1}{\baccept{m}} \left(\sum_{i=1}^{m} \left(\baccept{i}-\baccept{i-1}\right) \max\left(\frac{1 - \epsilon - \sqrt{\baccept{i-1}}}{1 + \beta}, 0 \right)\right) \\
    &\hspace{5pt}\geq \frac{1}{\baccept{m}} \int_0^{\baccept{m}}\max\left(\frac{1 - \epsilon - \sqrt{a}}{1 + \beta}, 0 \right) da \\
    &\hspace{5pt}\geq \int_0^{(1-\epsilon)^2}\frac{1 - \epsilon - \sqrt{a}}{1 + \beta}da \\&\hspace{5pt}\geq \frac{1-\epsilon}{3 (1 + \beta)}\,.
\end{align}
Where the inequalities come from the observation that the quantity being summed is a increasing function of $\baccept{i}$. Combining this bound with the lower bound on the accepting probability of the repeated blended measurement given in \Cref{thm:Blended_Quantum_OR} we see that, in Case 1, the probability that the repeated blended measurement accepts at least once and the first outcome it accepts corresponds on a measurement $M_i$ with $\Tr[M_i \rho] > 1-\epsilon$ is bounded below by
\begin{align}
    \frac{(1-\epsilon)}{3 (1 + \beta)} \frac{(1-\epsilon)^2}{4}
    \geq \frac{(1 -\epsilon)^3}{12 (1 + \beta)}\,,
\end{align}
as claimed. 
\end{proof}

\begin{proof}[Proof (\Cref{thm:random_event_finding}):]

The proof of the upper bound in Case 2 follows immediately from \Cref{thm:Random_Quantum_OR}.

To prove the lower bound in Case 1 we relate the measurement accepting probabilities in the random and blended measurement procedures. We use the same notation as in \Cref{sec:random_measurements}. Let $\state{i}$ be the state of a quantum system, initially in state $\rho$, after $i$ random measurements drawn with repetition from $\cM$ all reject and let $\bstate{i}$ be the state of a quantum system with the same initial state after $i$ applications of the blended measurement $\blend(\cM)$ all reject. Also relabel the measurements $M_i \in \cM$ so that 
\begin{align}
\Tr[M_i \rho] \geq 1-\epsilon
\end{align}
if and only if $M_i \leq k$. Let $\breject{i}$ be the probability that the blended measurement $\blend(\cM)$ is applied $i$ times to a quantum system in state $\rho$ and rejects each time, and $\reject{i}$ be the probability that $i$ random measurements from $\cM$ are applied to a quantum system initially in state $\rho$ and all reject. Finally let $\goodFirstAccept{i}$ be the event that the the first random measurement to accept accepts on the $i$'th round of the procedure, and the accepting measurement is a measurement $M_j$ with $j \leq k$.

Applying 
\Cref{cor:bounded_measurement_outcomes} with $X = \sum_{i\leq k} M_i$ gives 
\begin{align}
\reject{j} \sum_{i \leq k} \Tr[M_i \state{j}] \geq \breject{2j} \sum_{i \leq k} \Tr[M_i \bstate{2j}]
\end{align}
and hence 
\begin{align}
    \p[\goodFirstAccept{j}] \geq \p[\goodFirstAcceptB{2j}]\,.
\end{align}
We are interested in lower bounding 
\begin{align}
    \sum_{j=0}^m \p[\goodFirstAccept{j}] \geq \sum_{j=0}^m\p[\goodFirstAcceptB{2j}]\,.
\end{align}
The key observation is that 
\begin{align}
&\sum_{j=0}^m\p[\goodFirstAcceptB{2j}] \\
&\hspace{50pt}= \Tr[ \left( \sum_{j=0}^{m} E_0^{j} \left(m^{-1} \sum_{i\leq k} M_i\right) E_0^j\right) \rho ] \\
&\hspace{50pt}\geq \Tr[ \left( \sum_{j=0}^{m} E_0^{j} \left(m^{-1} \sum_{i\leq k} M_i\right) E_0^j\right) E_0^{1/2} \rho E_0^{1/2} ] - \sqrt{1 - \Tr[E_0 \rho]} \\
&\hspace{50pt}= \sum_{j=0}^m \p[\goodFirstAcceptB{2j+1}] - \sqrt{\p[\firstAcceptB{0}]}\,.
\end{align}
where the inequality follows from the gentle measurement lemma and the observation that 
\begin{equation}
    m^{-1} \sum_{i\leq k} M_i \leq 1 - E_0
\end{equation}
and hence 
\begin{equation}
0 \leq \sum_{j=0}^{m} E_0^{j} \left(m^{-1} \sum_{i\leq k} M_i\right) E_0^{j} \leq \frac{1}{1+E_0} \leq 1\,.
\end{equation}
Now we note that
\begin{align}
\p[\goodFirstAccept{0}] &= \p[\goodFirstAccept{0} | \firstAccept{0}] \p[\firstAccept{0}] \\&\geq \frac{k(1-\epsilon)}{k + \beta}\p[\firstAccept{0}] \\
&\geq \frac{1-\epsilon}{1 + \beta}\p[\firstAccept{0}] = \frac{1-\epsilon}{1 + \beta}\p[\firstAcceptB{0}]
\end{align}
so 
\begin{align}
\sum_{j=0}^m\p[\goodFirstAccept{j}]
&\geq \frac{1}{2} \sum_{j=0}^{2m}\p[\goodFirstAcceptB{j}] - \left(\frac{1 + \beta}{1 - \epsilon}\p[\goodFirstAccept{0}]\right)^{1/2}\,. 
\end{align}
Finally, we put this all together to obtain 
\begin{align}
&\frac{1}{2} \sum_{j=0}^{2m}\p[\goodFirstAcceptB{j}] \\
&\hspace{20pt}\leq  \sum_{j=0}^{m} \p[\goodFirstAccept{j}] + \left(\frac{1 + \beta}{1 - \epsilon}\p[\goodFirstAccept{0}]\right)^{1/2} \\
&\hspace{20pt}\leq \sum_{j=0}^{m} \p[\goodFirstAccept{j}] + \left(\frac{1 + \beta}{1 - \epsilon}\right)^{1/2} \left(\sum_{i=0}^{m} \p[\goodFirstAccept{j}]\right)^{1/2}
\end{align}
which (solving the resulting quadratic equation) gives the bound 
\begin{align}
\left(\sum_{j=0}^m \p[\goodFirstAccept{j}]\right)^{1/2}
&\geq \frac{1}{2} \left(- \left(\frac{1 + \beta}{1 - \epsilon}\right)^{1/2} + \left(\frac{1 + \beta}{(1 - \epsilon)} + 2\sum_{j=0}^{2m}\p[\goodFirstAcceptB{j}]\right)^{1/2}\right)\,.
\end{align}
Defining $\xi := \frac{(1-\epsilon)}{1+\beta} \sum_{j=0}^{2m}\p[\goodFirstAcceptB{j}]$ we can rewrite this equation as
\begin{align}
\left(\sum_{j=0}^m \p[\goodFirstAccept{j}]\right)^{1/2} &\geq \left(\frac{1 + \beta}{4(1 - \epsilon)}\right)^{1/2}\left(\left(1 +  2 \xi \right)^{1/2} - 1 \right) \\
&\geq \left(\frac{1 + \beta}{4(1 - \epsilon)}\right)^{1/2} \frac{2 \xi}{3} 
\end{align}
using that $(1+x)^{1/2} - 1 \geq \frac{\sqrt{3}-1}{2}x \geq \frac{1}{3}x$ whenever $x \leq 2$ to produce the inequality in the final line. Then we have 
\begin{align}
\sum_{j=0}^m \p[\goodFirstAccept{j}] &\geq \frac{1 + \beta}{4(1 - \epsilon)} \left(\frac{2\xi}{3}\right)^2 \\
&= \frac{1 - \epsilon}{9(1 + \beta)} \left(\sum_{j=0}^{2m}\p[\goodFirstAcceptB{j}]\right)^2 \\
&\geq \frac{(1 -\epsilon)^{7}}{1296 (1 + \beta)^3}\,.
\end{align}
This completes the proof.
\end{proof}

\section{Quantum Threshold Search}\label{sec:threshold_search}
In \Cref{sec:quantum_event_finding}, we were able to show that in certain cases performing repeated blended or random measurements solves event finding which we defined to be the problem of identifying an event that has a high probability of accepting on an unknown state, assuming such a measurement exists. Specifically, we showed the random and blended measurement procedures we gave in that section had success probability that scaled inverse polynomially (cubically) with the total accepting probability of all undesired measurement outcomes. In this section we present a finer-grained version of that analysis. We show that the average accepting probability of the measurement returned by the blended (or random) measurement procedure is related cubically to the average accepting probability of the measurement returned by repeatedly applying (random selected) measurements to fresh copies of the unkown state and returning the first measurement that accepts. This (in combination with Markov's inequality) implies the event finding result discussed earlier, but is a strictly more general result. In the later part of this section, we show we can use this stronger event finding lemma to improve our blended measurement protocol to a \emph{threshold search} protocol which uses $O(\log^2(m))$ copies of the unknown state. 


\subsection{A Stronger Event Finding Lemma}

We will state the stronger event finding lemma in terms of averages $\gamma$, and $\wtgamma^{\cB}_{j}$, which we now define.  Fix a set $\cM = \{M_i\}_{i \in [m]}$ of measurements that one wants to perform event finding over, and let $\rho \in \mathbb{C}^{d \times d}$ be an unknown quantum state.  Then define the following quantity (that depends implicitly on $\mathcal{M}$).
\begin{equation}\label{eqn:gamma}
    \gamma = \frac{\sum_{i \in [m]} \Tr[M_i \rho]^2}{\sum_{i \in [m]} \Tr[M_i \rho]}\,.
\end{equation}
To gain some intuition for $\gamma$, consider the following two-step procedure: (1) select a measurement $M_i$ at random from $\cM$ and apply it to $\rho$; (2) return $M_i$ and success if the measurement accepts and return failure otherwise. The quantity $\gamma$ is the average accepting probability of the measurement returned by this procedure conditioning on success. Equivalently, $M_i$ is the average accepting probability of the measurement obtained by selecting measurements at random from $\cM$, applying each of them to a fresh copy of $\rho$ and then returning the first measurement which accepts. Note that if the measurements $M_1, \ldots, M_m$ all commute (i.e. the situation is classical) and we apply them all to $\rho$ and return a random accepting measurement the average accepting probability of the measurement returned will also be $\gamma$.

Recall the notation used in the analysis of repeated blended measurements.  As in those sections, let $\rhobreject{j} = E_0^{j}\rho E_0^{j} / \Tr[E_0^{2j} \rho]$ be the post-measurement state after measuring the blended measurement $\cB(\cM)$ $j$ times and only seeing rejects (recall that $E_0$ is a Kraus operator, defined as the square root of the rejecting measurement).  
We are interested in the average accepting probability of the measurement returned by the $m$ round blended event finding procedure (in the event the blended event finding procedure does not return a measurement we say it has returned a measurement with accepting probability zero). Denoting this quantity by  $\wtgamma^{\cB}$ we see we can write it as. 
\begin{align}\label{eqn:wtgamma_blended}
    \wtgamma^{\cB} &= \sum_{i=1}^m \Tr[M_i \rho] \left(\sum_{j = 0}^{m-1} \frac{\Tr[M_i E_0^{j} \rho E_0^{j}]}{m} \right) \\
    &= \sum_{i = 1}^m \sum_{j = 0}^{m-1}
    \Tr[M_i \rho]\left(\frac{\Tr[M_i E_0^{j} \rho E_0^{j}]}{m} \right)
    \\
    &= \sum_{i = 1}^{m} \sum_{j = 0}^{m-1} (1 - \baccept{j}) \frac{\Tr[M_i\rhobreject{j}]\Tr[M_i\rho]}{m}\,.
\end{align}
We now show that this quantity is lower bounded by a polynomial function of $\gamma$. 

\begin{lem}
\label{lem:strong_event_finding}
    Let $\cM$ be a set of measurements, and $\rho$ be an unknown quantum state.  Let $\gamma$ and $\wtgamma^{\cB}$ be defined as in \Cref{eqn:gamma} and \Cref{eqn:wtgamma_blended}.   Then the following inequality holds:
    \begin{equation}
         \wtgamma^{\cB} \geq \frac{\gamma^3}{8}(1- \gamma^2/4)\,.
    \end{equation}
\end{lem}
\begin{proof}
Fix a value of $j$ and consider the following equation:
\begin{align}
    &\sum_{i = 1}^{m} (1 - \baccept{j}) \frac{\Tr[M_i\rhobreject{j}]\Tr[M_i\rho]}{m} \\&\hspace{60pt}\geq \sum_{i = 1}^{m} (1 - \baccept{j})\frac{\Tr[M_i \rho]}{m} \left(\Tr[M_i\rho] - \sqrt{\baccept{j}}\right)\\
    &\hspace{60pt}=(1 - \baccept{j})\left(\gamma \sum_{i = 0}^{m} \frac{\Tr[M_i \rho]}{m} - \sqrt{\baccept{j}}\sum_{i = 1}^{m}\frac{\Tr[M_i \rho]}{m} \right)\\
    &\hspace{60pt}=\baccept{0} (1 - \baccept{j})\left(\gamma - \sqrt{\baccept{j}}\right)\,.
\end{align}
In going from the first to second line, we apply the gentle measurement lemma and the operational definition of the trace distance.  We then use the definition of $\baccept{0}$, which is the denominator in the definition of $\gamma$.  Since we are trying to lower bound the sum, let $m^*$ be the smallest index such that $\baccept{m^*} \geq \gamma^2 / 4$.  By \Cref{thm:Blended_Quantum_OR}, this index is guaranteed to exist.  Then taking the sum up to $m^*$, we get
\begin{align}
    \wtgamma^{\cB} &\geq \sum_{j = 0}^{m^*-1}\sum_{i = 1}^{m} (1 - \baccept{j}) \frac{\Tr[M_i\rhobreject{j}]\Tr[M_i\rho]}{m}\\
    &\geq \sum_{j = 0}^{m^*-1}\baccept{0}\left(\gamma - \sqrt{\baccept{j}}\right) (1 - \baccept{j})\\
    &\geq m^* \baccept{0} \left(\gamma - \gamma/2\right)(1- \gamma^2/4)\\
    &\geq \left(\frac{\gamma}{2}\right)^2 \left(\frac{\gamma}{2}\right) (1- \gamma^2/16)\\
    &= \frac{\gamma^3}{8} (1- \gamma^2/4)\,.
\end{align}
We first substitute the lower bound for each individual $j$.  Then we use the fact that for all $j < m^*$, $\baccept{j} < \gamma^2/4$ to turn $\gamma - \sqrt{\baccept{j}}$ into $\gamma / 2$.  Finally we use the fact that every $\baccept{i} \leq i\baccept{0}$ to go from $m^*\baccept{0}$ down to $\baccept{m^*}$, and then the assumption that $\baccept{m^*} \geq \gamma^2/4$.    This completes the proof.
Note that $\gamma \leq 1$, so we could take $(1 - \gamma^2/4) \geq 3/4 \geq 1/2$ to get a cleaner looking lower bound of $\gamma^3 / 16$.
\end{proof}

Similar to the blended case, we can prove a strong event finding lemma for sequences of random measurements.  Recall that $\rhoreject{j}$ is the post-measurement state after applying $j$ many random measurements from $\cM$ and only seeing rejects.  We similarly define $\wtgamma$, the average accepting probability of the measurement returned by $m$ rounds of random measurements.  
\begin{align}
    \wtgamma &= \sum_{i = 1}^{m} \Tr[M_i \rho]\left(\sum_{j = 0}^{m-1} \frac{\sum_{T \in \mm{\notcM}{j}}\Tr[M_i T\rho T^{\dagger}]}{m^{j+1}}\right)\label{eqn:wtgamma_random}\\
    &= \sum_{j = 0}^{m-1} \left(\frac{\sum_{i = 1}^{m} \Tr[M_i \rho]\sum_{T \in \mm{\notcM}{j}}\Tr[M_i T\rho T^{\dagger}]}{m^{j+1}}\right)\\
    &= \sum_{j = 0}^{m-1} (1 - \accept{j})\left( \frac{\sum_{i = 1}^{m} \Tr[M_i \rho]\Tr[M_i \rhoreject{j}] }{m}\right)\,.
\end{align}
\begin{lem}\label{lem:strong_event_finding_random}
    Let $\cM$ be a set of projective measurements and $\rho$ be an unknown state.  Let $\gamma$ and $\wtgamma$ be as defined in \Cref{eqn:gamma} and \Cref{eqn:wtgamma_random}.  Then the following inequality holds:
    \begin{equation}
        \wtgamma \geq \frac{\gamma^{3}}{16}(1 - \gamma^2/4)\,.
    \end{equation}
\end{lem}
\begin{proof}
    Similar to before, fix a value of $j$ and consider the following
    \begin{align}
        &\sum_{i = 1}^{m}(1 - \accept{j})\frac{ \Tr[M_i \rho]\Tr[M_i \rhoreject{j}]}{m}\\ &\hspace{40pt}\geq \sum_{i = 1}^{m}(1 - \baccept{2j})\frac{ \Tr[M_i \rho]\Tr[M_i \rhobreject{2j}]}{m}\\
        &\hspace{40pt}\geq \baccept{0}(1 - \baccept{2j}) (\gamma - \sqrt{\baccept{2j}})\,.\label{eqn:wtgamma_single_term_lower_bound}
    \end{align}
    Here we use \Cref{cor:random_blended_trace_bound} in the first line, and proceed the same as in the proof of \cref{lem:strong_event_finding}, but with $2j$ instead of $j$.  Now let $m^*$ be the smallest index such that $\baccept{2m^*} \geq \gamma^2 / 4$.  Then we have that
    \begin{align}
        \wtgamma &\geq \sum_{j = 0}^{m^*-1} \baccept{0}(1 - \baccept{2j}) (\gamma - \sqrt{\baccept{2j}})\\
        &\geq m^*\baccept{0}\left(\gamma - \sqrt{\baccept{2j}}\right)(1 - \baccept{2j})\\
        &\geq m^*\baccept{0}(\gamma/2)(1 - \gamma^2/4)\\
        &\geq \left(\frac{1}{2}\right)\left(\frac{\gamma^2}{4}\right)\left(\frac{\gamma}{2}\right)(1 - \gamma^2/4)\\
        &\geq \frac{\gamma^3}{16}(1 - \gamma^2/4)\,.
    \end{align}
    Here we first use the lower bound we computed in \Cref{eqn:wtgamma_single_term_lower_bound}.  Then we use the fact that for all $j \leq m^*$, $\sqrt{\baccept{2j}} \geq \gamma/2$, and since no terms depend on $j$ after that we turn the sum into a factor of $m^*$.  Then we use the fact that $2m^*\baccept{0} \geq \baccept{2m^*} \geq \gamma^2/4$ and reduce the equation to its final form.  Since $\gamma \leq 1$, we will often use the slightly weaker lower bound of $\gamma^3 / 32$.  
\end{proof}

\subsection{Threshold Search via Repeated Blended Measurements}\label{sec:blended_threshold_search}
In the previous section, we strengthened the quantum event finding lemma
A related problem to quantum event finding is \emph{threshold search}~\cite{badescu2021improved}.  In threshold search, one is asked to output a measurement satisfying $\Tr[M_i \rho] \geq 1/3$ with constant probability, given the fact that there exists a measurement with $\Tr[M_i \rho] \geq 3/4$.  Ref~\cite{badescu2021improved} show that if this problem
can be solved for projective measurements using $k$ copies of the unknown state, then there is an algorithm for Shadow Tomography~\cite{aaronson2019shadow} that uses
\begin{equation}
    O\left(\frac{\log(d)L(k + L)}{\epsilon^4}\right)
\end{equation}
copies of the unknown state, where $L = \log\left(\frac{\log d}{\delta\epsilon}\right)$.  In this section, we show that the repeatedly applying blended or random measurements, while ``boosting'' around a uniformly random threshold, solves threshold search using $O(\log^2(m))$ copies of the unknown state, matching the best known shadow tomography upper bounds from Ref.~\cite{badescu2021improved}.

We first introduce the \emph{binomial measurement}, which boosts the sensitivity of a measurements to values far from a threshold $\theta$.  
\begin{lem}\label{lem:binomial_measurement}
    Let $\rho \in \mathbb{C}^{d \times d}$ be a quantum state and $A$ be a quantum event.  Let $n \in \mathbb{N}$ and $\theta \in [0, 1]$ be an arbitrary threshold, and let $S \sim \Binomial(n, \Tr[A\rho])$.  Then there exists a quantum event $\bin{A}{n}{\theta} \in (\mathbb{C}^{d \times d})^{\otimes n}$ such that 
    \begin{equation}
        \Tr[\bin{A}{n}{\theta}\rho^{\otimes n}] = \p[S \geq \theta n]\,.
    \end{equation}
    We call the measurement $\bin{A}{n}{\theta}$ the \emph{binomial measurement of $A$ over $n$ registers with threshold $\theta$}. If $A$ is projective then so is $\bin{A}{n}{\theta}$.
\end{lem}
\begin{proof}
    Define $A_{1} = A$ and $A_{0} = 1 - A$.  For $x \in \{0, 1\}^{n}$, define $A_{x}$ to be $\bigotimes_{i = 1}^{n} A_{x_i}$.  Then for $0\leq k \leq n$, define
    \begin{equation}
        E_{k} = \sum_{x\in\{0, 1\}^{n}: |x| = k} A_{x}\,.
    \end{equation}
    Here $|x|$ denotes the Hamming weight of $x$.  Finally define the measurement
    \begin{equation}
        \bin{A}{n}{\theta} = \sum_{k \geq \theta n} E_{k}\,.
    \end{equation}
    Recalling that $S \sim \Binomial(n, \Tr[A\rho])$, it is clear that for all $k$, $\Tr[E_{k}\rho^{\otimes n}] = \p[S = k]$, so $\Tr[\bin{A}{n}{\theta} \rho^{\otimes n}] = \p[S \geq \theta n]$. Finally, we note that if $A$ is projective then $\bin{A}{n}{\theta}$ is a sum of orthogonal projectors, and so is also a projector. 
\end{proof}

Given a set of measurement $\cM$ for threshold search and a threshold $\theta$, we define the set of binomial measurements with threshold $\theta$ using $n$ copies, $\bin{\cM}{n}{\theta} = \{\bin{M_i}{n}{\theta}: i \in [m]\}$.  We can similarly define the blended measurement corresponding to $\bin{\cM}{n}{\theta}$.  Consider the following algorithm for threshold search.

\begin{framed}
\begin{algorithm}
\label{algorithm:blended_measurement_threshold_search}
\textbf{Blended Measurement Threshold Search}
\end{algorithm}

\noindent \textbf{Input:} A classical description of a set of two outcome measurements $\cM = \{M_1, M_2, \ldots , M_m\}$ and $n = 100\log^2 m$ copies of an unknown state $\rho$. \mbox{} \vspace{3mm} \newline
\textbf{Output:} Measurement $M_i$ or \REJECT. 
\begin{enumerate}
    \item Select a random $\theta \in [2/5, 3/5]$.
    \item Repeat $m$ times:
    \begin{enumerate}
        \item Perform the blended measurement $\cB(\bin{\cM}{n}{\theta})$ on $n$ copies of the state.  If the measurement accepts, return the accepting outcome.
    \end{enumerate}
    \item Return \REJECT.
\end{enumerate}
\end{framed}
In order to show that the algorithm works, we first show that if there is a measurement satisfying $\Tr[M_i \rho] \geq 3/4$ and $\wtgamma^{\cB} \geq c$ (for any constant $c$) for a fixed choice of $\theta$, the algorithm outputs a measurement satisfying $\Tr[M_i \rho] \geq 1/3$ with probability at least $c - 1/m$.  Then, we show that with constant probability (over the choice of the threshold $\theta$), $\gamma \geq 1/8$, which implies a constant lower bound for $\wtgamma^{\cB}$ by \Cref{lem:strong_event_finding}.

\begin{lem}\label{lem:constant_success_probability}
   Fix a choice of threshold $\theta$, let $\wtgamma^{\cB}$ be defined as in \Cref{eqn:wtgamma_blended} with respect to $\bin{\cM}{100\log^2(m)}{\theta}$.  Then if there exists a constant $c$ satisfying 
   \begin{equation}
       \wtgamma^{\cB} \geq c\,,
   \end{equation}
   Then \cref{algorithm:blended_measurement_threshold_search} outputs a measurement satisfying $\Tr[M_i\rho] \geq 1/3$ with probability at least $c - 1/m$.  
\end{lem}
\begin{proof}
    $\wtgamma^{\cB}$ is the average value of $\Tr[\bin{M_i}{n}{\theta} \rho^{\otimes n}]$ over the output distribution of the blended measurement threshold search algorithm.  For all measurements satisfying $\Tr[M_i \rho] < 1/3$, we know that the binomial measurement satisfies $\Tr[\bin{M_i}{n}{\theta} \rho^{\otimes n}] < 1/m$.  Let $p_b$ be the probability of outputting a measurement satisfying $\Tr[M_i \rho] < 1/3$, and $p_g$ be the probability of outputting a measurement satisfying  $\Tr[M_i \rho] \geq 1/3$.  Then we have
    \begin{equation}
        p_g \geq c - p_b/m \geq c- 1/m\,.
    \end{equation}
    Thus the probability of outputting a measurement satisfying $\Tr[M_i \rho] \geq 1/3$ is at least $c - 1/m$.
\end{proof}

Up to this point, we have shown $\wtgamma^{\cB}$ is lower bounded by a function of $\gamma$, and that having a constant lower bound for $\wtgamma^{\cB}$ implies a constant success probability for the algorithm.  All that remains to be seen is that $\gamma$ is constant for most choices of the threshold $\theta$.  In order to show that most thresholds $\theta$ are ``good'', we define some functions and sets that are going to be used in heavily.  Given two numbers, $\alpha$ and $\beta$, define the set $\cM[\alpha, \beta]$ to be as follows
\begin{equation}
    \cM[\alpha, \beta] = \{i : \alpha \leq \Tr[M_i \rho] \leq \beta\}\,.
\end{equation}
Define $n(\alpha, \beta)$ to be the size of $\cM[\alpha, \beta]$.  We say that a value of $\theta$ is ``bad'' if the following holds
\begin{equation}\label{eqn:theta_bad}
    \sum_{i \in \cM[0, \theta]} \exp(-100 \log^2 m (\theta - \Tr[M_i\rho])^2) \geq 4n(\theta, 1)\,.
\end{equation}
The left hand side is related to the probability that a measurement with accepting probability below $\theta$ will be chose and accepted by the binomial measurement (using $k = 100 \log^2m$ copies of the state), and the right side is related to the probability that a measurement with accepting probability higher than $\theta$ will be chosen. ``Good'' threshold values are those that are not bad. 

The following claims show that the set of bad measurements has measure bounded by a constant.
\begin{claim}\label{claim:bad_threshold_interval_existance}
    Assume that for the set of measurements $\cM$, $n(\theta, 1) \geq 1$.  Let $\theta$ be a bad threshold, then there exists a number $\beta_{\theta} \leq \theta$ such that
    \begin{equation}
        n(\beta_{\theta}, \theta) \geq \exp(50 \log^2 m (\theta - \beta_{\theta})^2) n(\theta, 1)\,.
    \end{equation}
\end{claim}
\begin{proof}
    Assume for the sake of contradiction that for all $\beta \leq \theta$,
    \begin{equation}
        n(\beta, \theta) \leq \exp(50 \log^2 m(\beta - \theta)^2)\,.
    \end{equation}
    We want to arrive at the contradiction with the fact that $\theta$ is a bad threshold.  To do so, we can evaluate the left hand side of \Cref{eqn:theta_bad}.  For ease of notation, let $\eta(x) = n(\theta - x, \theta)$, $f(x) = \exp(-100 \log^2(m) x^2)$, and $f'(x) = \frac{d}{dx} f(x)$.  Also let $L = n(0, \theta)$, and denote by $(y_i)_{i = 1}^{L}$ the list of $\Tr[M_i \rho]$ for $i \in \cM[0, \theta]$, in increasing order.  Let $x_i = \theta - y_i$.  Then we have that
    \begin{align}
        \sum_{i \in n(0, \theta)} \text{exp}(-100\log^2 m (\theta - \Tr[M_i \rho])^2) &= \sum_{i = 1}^L f(\theta - y_i) \\
        &= \sum_{i = 1}^L f(x_i) \\
        &= \sum_{i = 1}^{L - 1} -i (f(x_{i+1}) - f(x_i)) + L f(x_{L}) \\
        &= \sum_{i=0}^{L-1} -i \left( \int_{x_i}^{x_{i+1}} f'(x) dx \right) - \int_{x_L}^\infty L f'(x) dx \label{eq:sum_to_int} \\
        &=  - \sum_{i=0}^{L-1}\left( \int_{x_i}^{x_{i+1}}  \eta(x) f'(x) dx \right) - \int_{x_L}^\infty n(x) f'(x) dx \label{eq:count_to_cdf}\\
        &= - \int_{x_1}^\infty \eta(x) f'(x) dx \\
        &= - \int_0^\infty \eta(x) f'(x) dx \,.\label{eq:extending_integral}
    \end{align}
    Note that for all $x > 0$ we have
    \begin{align}
        f'(x) = -200 \log^2(m) x \exp( - 100 \log^2(m) (x)^2) \leq 0\,.
    \end{align}
    We can plug in the assumed upper bound on $n(\beta, \theta) = \eta(\theta - \beta)$ to obtain 
    \begin{align}
        &\sum_{i \in n(0, \theta)} \text{exp}(-100\log^2 m (\theta - \Tr[M_i \rho])^2) \nonumber \\
        &\hspace{60pt}\leq \int_{0}^{\infty} \exp(50 \log^2 m x^2) (200x\log^2 m)\exp(-100 \log^2 m x^2) dx\\
        &\hspace{60pt}= \int_{0}^{\infty} (200x\log^2 m)\exp(-50 \log^2 m x^2) dx\\
        &\hspace{60pt}\leq -4 \exp(-50 \log^2 m x^2)\bigg\rvert_{0}^{\infty}\\
        &\hspace{60pt}= 4\\
        &\hspace{60pt} \leq 4 n(\theta, 1)\,.
    \end{align}
    Here the last line uses the assumption that $n(\theta, 1) \geq 1$.  This contradicts \Cref{eqn:theta_bad}, which proves that a suitable choice of $\beta_{\theta}$ exists.  
\end{proof}

Now we know that for every bad choice of threshold $\theta$, there is an interval below $\theta$ in which the number of measurements that have accepting probabilities lying in the interval is exponentially large, compared to the size of the interval.  We recursively define the following sets of intervals, which contain all of the bad thresholds.  Let $\theta_0 = \max\{\theta \leq 2/3 : \theta \text{ bad}\}$.
\begin{equation}
    \theta_i = \max\{\theta \leq \beta_{\theta_i-1} : \theta \text{ bad}\}\,.
\end{equation}
For convenience, define $\beta_i = \beta_{\theta_i}$, and let $d_i = \theta_i - \beta_i$.  Let $N$ be the total number of $\theta_i$.  First, we show that the total size of the intervals (squared) lower bounds the number of measurements contained in intervals.  
\begin{claim}
    Assume that $n(2/3, 1) \geq 1$, then the following bound holds on the number of measurements contained in the intervals.
    \begin{equation}
        \sum_{i = 0}^{N} n(\beta_i, \theta_i) \geq \exp\left(50 \log^2(m) \sum_{i} d_i^2\right)\,.
    \end{equation}
\end{claim}
\begin{proof}
    We begin by induction. By \Cref{claim:bad_threshold_interval_existance} and the assumption that $n(2/3, 1) \geq 1$ (and therefore $n(\theta_0, 1) \geq 1$, we have that
    \begin{equation}
        n(\beta_0, \theta_0) \geq \exp(50 \log^2(m) d_0^2)\,.
    \end{equation}
    Now assume that the claim holds for the sum from $0$ to $j-1$, we evaluate the sum up to $j$ as follows
    \begin{align}
        \sum_{i = 0}^{j} n(\beta_i, \theta_i) &= n(\beta_j, \theta_j) + \sum_{i = 0}^{j-1} n(\beta_i, \theta_i)\\
        &\geq \exp(50 \log^2(m) d_j^2)n(\theta_j, 1) + \sum_{i = 0}^{j-1} n(\beta_i, \theta_i)\\
        &\geq (1 + \exp(50 \log^2(m) d_j^2)\sum_{i = 0}^{j-1} n(\beta_i, \theta_i)\\
        &\geq \sum_{i = 0}^{j} n(\beta_i, \theta_i) \\
        &\geq \exp\left(50 \log^2(m) \sum_{i} d_i^2\right)\,.
    \end{align}
    Going to the second line we use \Cref{claim:bad_threshold_interval_existance}, and to get to the third line we use the fact that the interval from $\theta_j$ to $1$ contains all of the previous intervals, and the definition of $n$.  
\end{proof}

We now show that the number of intervals is bounded by $\log(m)$.  
\begin{claim}
    If $n(2/3, 1) \geq 1$, then $N \leq \log(m)$.  
\end{claim}
\begin{proof}
    In order to prove the claim, we show by induction that for all $j$
    \begin{equation}
        \sum_{i = 0}^{j} n(\beta_i, \theta_i) \geq 2^{j}\,.
    \end{equation}
    By assumption the assumption that $n(2/3, 1) \geq 1$ and \Cref{claim:bad_threshold_interval_existance}, we have that $n(\beta_0, \theta_0) \geq 1$.  Now assume that $\sum_{i = 0}^{j-1} n(\beta_i, \theta_i) \geq 2^{j-1}$.  By \Cref{claim:bad_threshold_interval_existance}, for every $j$, $n(\beta_j, \theta_j) \geq n(\theta_j, 1) \geq \sum_{i = 0}^{j-1} n(\beta_i, \theta_i)$.  So, we have that
    \begin{equation}
        \sum_{i = 0}^{j} n(\beta_i, \theta_i) \geq 2 \sum_{i = 0}^{j-1} n(\beta_i, \theta_i) \geq 2^{j}\,.
    \end{equation}
    Because we have $m$ measurements in total, we must have that $N \leq \log(m)$.
\end{proof}

\begin{lem}\label{lem:bad_thresholds_upper_bound}
    Assume that $n(\theta, 1) \geq 1$.  Then the set of bad thresholds has measure less than $1/6$.  
\end{lem}
\begin{proof}
    By definition, every bad threshold is contained in some interval $[\beta_i, \theta_i]$, so to upper bound the measure of bad thresholds, it suffices to upper bound $\sum_{i}^{N} d_i$.  The value of the following optimization problem is an upper bound on the measure of bad thresholds.
    \begin{align}
    \max \sum_{i = 0}^{N} d_i \\
    \text{subject to } N &\leq \log(m)\,, \nonumber\\
    \sum_{i = 0}^{N} d_i^2 &\leq \frac{1}{50\log(m)}\,,\nonumber\\
    d_i &\geq 0 \quad \forall i\,,\nonumber\\
    N &\geq 0\,.\nonumber
\end{align}
The best value one can achieve for this optimization problem occurs when $N = \log(m)$ and 
\begin{equation}
    d_1 = d_2 = \ldots = d_N = \sqrt{\frac{1}{50 \log^2(m)}}\,.
\end{equation}
Computing the sum of $d_i$, we get that the measure of bad thresholds is upper bounded by
\begin{align}
    \sum_{i}^{N} d_i &\leq \log(m) \sqrt{\frac{1}{50 \log^2(m)}} \\
    &\leq \sqrt{\frac{1}{50}}\\
    &\leq 1/6\,.
\end{align}
This completes the proof.
\end{proof}
Before proving that the theorem works, we show that a measurement not being bad implies a lower bound for the value of $\gamma$.  This implies that if the algorithm chooses a good threshold, it has a constant success probability for the threshold search problem.
\begin{claim}
    If $\gamma_{\theta} < 1/32$, then $\theta$ is bad.
\end{claim}
\begin{proof}
    Letting $n = 100\log^2m$, the following is a lower bound for $\gamma$ by definition
    \begin{equation}
        \frac{\sum_{i \in \cM[\theta, 1]} \Tr[\bin{M_i}{n}{\theta}\rho^{\otimes n}]^2}{\sum_{i} \Tr[\bin{M_i}{n}{\theta}\rho^{\otimes n}]} \leq \gamma \leq 1/32\,.
    \end{equation}
    The denominator can be re-written as the sum over $\cM[\theta, 1]$ and $\cM[0, \theta]$.  Doing this and rearranging terms, we get that
    \begin{multline}
        \frac{1}{32} \sum_{i \in \cM[0, \theta]} \Tr[\bin{M_i}{n}{\theta}\rho^{\otimes n}] \\\geq \sum_{i \in \cM[\theta, 1]} \Tr[\bin{M_i}{n}{\theta} \rho^{\otimes n}]^2 - \frac{1}{32} \sum_{i \in \cM[\theta, 1]} \Tr[\bin{M_i}{n}{\theta} \rho^{\otimes n}]\,.
    \end{multline}
    Using the fact that for all indices $i \in \cM[\theta, 1]$, we have that $1/2 \leq \Tr[\bin{M_i}{n}{\theta}\rho^{\otimes k}] \leq 1$, so the first term on the right side is greater than $n(\theta, 1)/4$, and the second term is less than $n(\theta, 1)/32$, we get the following
    \begin{equation}
        \frac{1}{32}\sum_{i \in \cM[0, \theta]} \Tr[\bin{M_i}{n}{\theta}\rho^{\otimes n}] \geq \frac{7}{32} n(\theta, 1)\,.
    \end{equation}
    Finally, by the Chernoff bound on the binomial distribution, setting $k = 100 \log^2 m$ this is a lower bound the left hand side of \Cref{eqn:theta_bad}, showing that $\theta$ is bad. 
\end{proof}

Note that the contrapositive of the previous claim is that if $\theta$ is good, then it must be that $\gamma_{\theta} \geq 1/32$.  Now we prove that the blended measurement threshold search solves threshold search.
\begin{thm}[Blended Measurement Threshold Search]\label{thm:blended_measurement_threshold_search}
    Let $\cM$ be a set of measurements and $\rho$ be an unknown quantum state.  Assume that there is a measurement in $\cM$ satisfying $\Tr[M_i \rho] \geq 3/4$.  Then blended measurement threshold search (\cref{algorithm:blended_measurement_threshold_search}) returns a measurement satisfying $\Tr[M_i \rho] \geq 1/3$ with constant probability.  
\end{thm}
\begin{proof}
    By assumption, there exists at least one measurement satisfying $\Tr[M_i \rho] \geq 3/4$, so by \Cref{lem:bad_thresholds_upper_bound}, the algorithm selects a good threshold with probability at least $1/5 - 1/6 = 1/30$.  If a good threshold is chosen, then $\gamma_{\theta} \geq 1/32$ from the previous claim.  By the strong event finding lemma (\Cref{lem:strong_event_finding}), and \Cref{lem:constant_success_probability}, the probability that the algorithm outputs a measurement satisfying $\Tr[M_i \rho] \geq 1/3$, conditioned on selecting a good threshold, is lower bounded by
    \begin{equation}
        \frac{\gamma^3}{16} - \frac{1}{m} \geq \frac{1}{32000}\,,
    \end{equation}
    for a suitably large choice of $m$.  Putting these two together, blended measurement threshold search succeeds in finding a measurement with high accepting probability with probability at least
    \begin{equation}
        \frac{1}{30}\cdot \frac{1}{32000} \geq 10^{-6}\,,
    \end{equation}
    which is a constant in $m$.
\end{proof}

\subsection{Threshold Search via Repeated Random Measurements}
In \Cref{sec:blended_threshold_search}, we showed that picking a threshold value to boost around uniformly at random is ``good'' with constant probability.  Leveraging those results, we can prove a sequence of random measurements also succeeds in performing threshold search.  Consider the following random measurement version of Algorithm $3$:

\begin{framed}
\begin{algorithm}
\label{algorithm:random_measurement_threshold_search}
\textbf{Random Measurement Threshold Search}
\end{algorithm}

\noindent \textbf{Input:} A classical description of a set of two outcome measurements $\cM = \{M_1, M_2, \ldots , M_m\}$ and $n = 100\log^2 m$ copies of an unknown state $\rho$. \mbox{} \vspace{3mm} \newline
\textbf{Output:} Measurement $M_i$ or \REJECT. 
\begin{enumerate}
    \item Select a random $\theta \in [2/5, 3/5]$.
    \item Repeat $m$ times:
    \begin{enumerate}
        \item Pick a random measurement $M_i \in \cM$.
        \item Perform the measurement $\bin{M_i}{n}{\theta}$ on the current state.  If the measurement accepts, output $M_i$.
    \end{enumerate}
    \item Return \REJECT.
\end{enumerate}
\end{framed}

The lower bound on $\gamma$ and the strong event finding lemma for repeated random measurements hold, so all that remains to be seen is that a constant lower bound on $\wtgamma$ implies a constant success probability.  
\begin{lem}\label{lem:constant_success_probability_random}
      Fix a choice of threshold $\theta$, let $\wtgamma$ be defined as in \Cref{eqn:wtgamma_random} with respect to $\bin{\cM}{n}{\theta}$.  Then if there exists a constant $c$ satisfying 
   \begin{equation}
       \wtgamma \geq c\,,
   \end{equation}
   Then \cref{algorithm:random_measurement_threshold_search} outputs a measurement satisfying $\Tr[M_i\rho] \geq 1/3$ with probability at least $c - 1/m$.   
\end{lem}
The proof of \Cref{lem:constant_success_probability} used nothing about the definition of $\wtgamma^{\cB}$ other than the fact that it was the average accepting probability of blended measurement threshold search.  Since $\wtgamma$ is the same quantity for random measurement threshold search, the same proof works, so we omit it here.

\begin{thm}[Random Measurement Threshold Search]\label{thm:random_measurement_threshold_search}
    Let $\cM$ be a set of projective measurements and $\rho$ be an unknown quantum state.  Assume that there is a measurement in $\cM$ satisfying $\Tr[M_i \rho] \geq 3/4$.  Then random measurement threshold search (\cref{algorithm:random_measurement_threshold_search}) returns a measurement satisfying $\Tr[M_i \rho] \geq 1/3$ with constant probability.  
\end{thm}
\begin{proof}
    Similar to \Cref{thm:blended_measurement_threshold_search}, we have that the probability the algorithm selects a good threshold is lower bounded by $1/30$, and if a good threshold is chosen, $\gamma_{\theta} \geq 1/32$.  By the strong event finding lemma (\Cref{lem:strong_event_finding_random}), now for random measurements, and \Cref{lem:constant_success_probability_random}, we have that the probability of outputting a measurement satisfying $\Tr[M_i \rho] \geq 1/3$ if a good threshold is chosen is lower bounded is
    \begin{equation}
        \frac{\gamma^3}{32} - 1/m \geq \frac{1}{64000}\,.
    \end{equation}
    for a suitably large choice of $m$.  Putting these together, we find that random measurement threshold search succeeds in finding a measurement satisfying $\Tr[M_i \rho] \geq 1/3$ with probability
    \begin{equation}
        \frac{1}{30} \cdot \frac{1}{64000} \geq \frac{1}{2 \cdot 10^{6}}\,,
    \end{equation}
    which is a constant in $m$.
\end{proof}
Note that we make no claims that the constant lower bounds are close to optimal.  It is likely with a careful eye and some effort, a more reasonable looking constant success probably could be proven.  We leave the task of tuning the constants for future work.

\subsection{Commentary on Shadow Tomography}\label{sec:shadow_tomography_commentary}
We have presented an algorithm for shadow tomography that uses (asymptotically) the same number of measurements as Ref.~\cite{badescu2021improved}. The algorithm presented in Ref.~\cite{badescu2021improved} was online, but required the introduction of Laplace noise in the measurements performed. The algorithm we present does not require Laplace noise, but is no longer online: while the random measurement version of threshold search we present here does make measurements ``one at a time'', it requires the measurements to be made in a random order, and we expect it to fail for adversarially ordered measurements. We note that the use cases (hypothesis selection and consequently simpler full state tomography) mentioned in Ref.~\cite{badescu2021improved} do not require that threshold search be online, and so can all be achieved with the same sample complexity using the algorithm presented here.  On the surface, it seems like nothing was lost or gained in showing that this algorithm worked and nothing new was enabled by it, so a natural question is why did the authors put effort into proving correctness of the new algorithm.  The answer (the authors believe) is that the idea of using random sequences of measurements (along with other bells and whistles) to solve shadow tomography seems like a promising approach to improving the best known upper bound, and a way to avoid existing lower bounds against shadow tomography.

Shadow tomography shares much in common with the study of differential privacy~\cite{aaronson2019gentle}.  Within the field of differential privacy, one way to get around lower bounds against online algorithms is to introduce randomness.  These trusted sources of randomness are referred to as ``shufflers''.  Intuitively, in the context of differential privacy the reason a shuffler might be expected to improve privacy is simple: An adversary looses control over part of a differentially private process, and so should not be able to learn as much as without the shuffler.  And indeed many tasks, such as uniformity testing~\cite{bassily2015local, cheu2021pure}, and perhaps more fittingly histogram estimation~\cite{bassily2015local, balcer2021connecting, balcer2019separating} have improved upper bounds in the shuffled model.  For quantum states, the connection is slightly more subtle.  In the context of shadow tomography there is no adversary, the algorithm has full control over the quantum system, and so one might hope for a deterministic algorithm for shadow tomography, similar to the settings of full state tomography~\cite{o2016efficient, o2017efficient} and classical shadows~\cite{grier2022sample}.  However, randomness does help quantum algorithms in one important way.  The information-disturbance trade-off implies that learning information about quantum states is damaging (formally exhibited by the gentle measurement lemma), and so the algorithm itself plays the role of the adversary.  Any information that the algorithm learns that does not help it solve the problem constitutes unnecessary damage to the algorithm's quantum state, so from the perspective of the algorithm it actually might make sense to delegate decisions like flipping random coins to an outside party to prevent the algorithm from learning unnecessary information.  The authors believe the connections to shuffled differential privacy, as well as the power of forgetting for quantum algorithms are good reasons to study random sequence of measurements in the context of shadow tomography, and potentially an avenue to improve on existing upper bounds. 

\section{Quantum Mean Estimation}
\label{sec:quantum_mean_estimation}

In this section we introduce a new Event Learning Problem in which the goal is to estimate the average accepting probability 
\begin{align}
    \frac{1}{\abs{\cM}} \sum_{M_i \in \cM} \Tr[M_i \rho]
\end{align}
of a set of measurements $\cM$ on an unknown state $\rho$.  We give a protocol based on blended measurements and compare its performance with an analogous classical algorithm.  

Given a set of two outcome measurements $\cM = \{M_1, M_2, \ldots , M_m\}$ define the matrices 
\begin{enumerate}
    \item $\avM = \frac{1}{m} \sum_{M \in \cM} M$
    \item $\avnotM= \frac{1}{m} \sum_{M \in \cM} (1-M)$.
\end{enumerate}
With this, we define the following mean estimation procedure. 

\begin{framed}
\begin{algorithm}
\label{algorithm:quantum_mean_estimation}
    \textbf{Quantum Mean Estimation}
\end{algorithm}

\noindent \textbf{Input:} A classical description of a set two outcome measurements $\cM = \{M_1, M_2, \ldots , M_m\}$ and $k$ copies of a $d$ dimensional state $\rho$. \mbox{} \vspace{3mm} \newline
\textbf{Output:} A estimate of the average accepting probability $\frac{1}{m} \sum_i \Tr[M_i \rho]$
\begin{enumerate}
    \item For each copy of the state $\rho$:
    \begin{enumerate}
    \item Prepare a quantum system in state $\rho$.
    \item Apply the two outcome measurement $\{ \avM^{1/2}, \avnotM ^{1/2} \}$ to the quantum system $t$ times in sequence and count the number of times this measurement accepts. Let $A_j$ denote the total number of accepts observed on the $j$'th quantum system. 
    \end{enumerate}
    \item Output the estimate $\sum_{j \in [k]} A_j/(tk)$.
\end{enumerate}
\end{framed}

We first show that the expected value of the estimator given by the above algorithm is unbiased.
\begin{thm}
The expected value of the estimate given by \cref{algorithm:quantum_mean_estimation} is
\begin{align}
    \Tr[\avM \rho] = \frac{1}{m} \sum_i \Tr[M_i \rho]\,.
\end{align}
\end{thm}
\begin{proof}
They key observation is that 
\begin{align}
    \Tr[\avM \left(\avM^{1/2} \rho \avM^{1/2} + \avnotM^{1/2} \rho \avnotM^{1/2} \right) ] = \Tr[\avM (\avM + \avnotM) \rho] = \Tr[\avM \rho]
\end{align}
and so applying the measurement $\{\avM^{1/2}, \avnotM^{1/2}\}$ to a system in state $\rho$ does not (in expectation over measurement outcomes) change the probability of subsequent $\{\avM^{1/2}, \avnotM^{1/2}\}$ measurements accepting (this is a general feature of two outcome measurements). Then the expected value of each $A_j$ is given by $t \Tr[\avM \rho]$ and the expected value of the final estimate is $\Tr[\avM \rho]$, as desired. 
\end{proof}
Next we determine the precision of the estimator output by Algorithm 3 by bounding its variance.   
\begin{lem} \label{lem:mean_estimation_variance}
Let $\{\ket{\psi_a}\}$ be the eigenvectors of $\avM$ with eigenvalues $\{\lambda_a\}$, and $\rho = \ketbra{\phi}$ where $\ket{\phi} = \sum_{a=1}^{d} \alpha_a \ket{\psi_a}$.  Then the variance of the estimator given in \cref{algorithm:quantum_mean_estimation} when $k = 1$ is
\begin{align}
        \frac{1}{t}\sum_{a=1}^{d} \alpha_a^2\lambda_l(1-\lambda_a) + \left(\sum_{a=1}^{d} (\alpha_a \lambda_a)^2 - \sum_{a, b = 1}^{d} \left(\alpha_a\alpha_b \lambda_a \lambda_b\right)^2\right)\,.
\end{align}
\end{lem}
\begin{proof}
Recall that $A_j$ denotes the number of accepting measurements seen in the process.  Then 
\begin{align}
    \p[A_1 = x] &= \binom{t}{x}\Tr(\avM^x \avnotM^{t-x} \rho)\\
    &= \binom{t}{x}\sum_{a=1}^{d}\alpha_a^2 \lambda_a^x(1-\lambda_a)^{t-x}\,.
\end{align}
Now, given a binomial distribution with $t$ coins and success probability $p$, $\mathbb{E}[\Binomial(t, p)^2]$ is given as follows
\begin{align}
    \mathbb{E}[\Binomial(t, p)^2] &= \sum_{k \geq 0}k^2 \binom{t}{x}p^x (1-p)^{t-x}\\
    &= t^2 p^2 + t p(1-p)\,.
\end{align}
Thus, we can compute the expectation of our distribution squared as follows
\begin{align}
    \mathbb{E}[A_1^2] &= \sum_{a=1}^{d}\alpha_a^2\left(\sum_{x \geq 0} x^2 \binom{t}{x} \lambda_a^x (1-\lambda_a)^{t-x}\right)\\
    &= \sum_{a=1}^d\alpha_a^2 \mathbb{E}[\Binomial(t, \lambda_a)^2]\\
    &= \sum_{a=1}^{d}\alpha_a^2 \left(t^2 \lambda_a^2 + t \lambda_a (1-\lambda_a)\right)\,.
\end{align}
From this we get the variance of the distribution.
\begin{align}
    \text{Var}[A_1] &= \mathbb{E}[A_1^2] - \mathbb{E}[A_1]^2\\
    &= \sum_{a=1}^{d}\alpha_a^2 \left(t^2 \lambda_a^2 + t \lambda_i (1-\lambda_a)\right) - \sum_{a, b = 1}^{d} \alpha_a^2 \alpha_b^2 t^2 \lambda_a^2 \lambda_b^2\\
    &= t \sum_{a=1}^{d} \alpha_a^2\lambda_a(1-\lambda_a) + t^2 \left(\sum_{a=1}^{d} (\alpha_a \lambda_a)^2 - \sum_{a, b = 1}^{d} \left(\alpha_a\alpha_b \lambda_a \lambda_b\right)^2\right)\,.
\end{align}
Our estimate of the $\Tr(\avM\rho)$ is given by $A_1/t$, which means the variance of our estimate is going to be given by
\begin{equation}
    \frac{1}{t}\sum_{a=1}^{d} \alpha_a^2\lambda_a(1-\lambda_a) + \left(\sum_{a=1}^{d} (\alpha_a \lambda_a)^2 - \sum_{a, b = 1}^{d} \left(\alpha_a\alpha_b \lambda_a \lambda_b\right)^2\right)\,.
\end{equation}
This completes the proof.
\end{proof}

\begin{rmk}\label{rmk:mean_estimation_variance}
When we take the limit as $t$ goes to infinity, the first term in the variance goes to $0$ and we are left with a residual variance of \begin{equation}
    \sigma^2 := \left(\sum_{a=1}^{d} (\alpha_a \lambda_a)^2 - \sum_{a, b = 1}^{d} \left(\alpha_a\alpha_b \lambda_a \lambda_b\right)^2\right) \label{eq:mean_estimation_variance} 
\end{equation}
Of course, the $\alpha_a$ are not known a-priori. Maximizing over possible values of $\alpha_a$ we find \Cref{eq:mean_estimation_variance} is upper bounded by $\frac{1}{4} (\lambda_{\text{max}} - \lambda_{\text{min}})^2$. (A proof of this fact is given by \Cref{lem:variance_upper_bound} in the Appendix.) 
\end{rmk}

\begin{cor}
There exists a protocol that uses $O(\sigma^2/\epsilon^2) \leq O((\lambda_\text{max} - \lambda_{\text{min}})^2/\epsilon^2)$ copies of an unknown state $\rho$ to estimate $\Tr(\avM\rho)$ to within error $\epsilon$ (with constant success probability). 
\end{cor}

\begin{proof}
From \Cref{lem:mean_estimation_variance} it follows that applying \cref{algorithm:quantum_mean_estimation} with $t = O(\sigma^{-2})$ and $k = 1$ produces an estimate with variance
\begin{align}
    O\left(\frac{1}{\sigma^{-2}}\sum_{a=1}^{d} \alpha_a^2\lambda_a(1-\lambda_a) + \left(\sum_{a=1}^{d} (\alpha_a \lambda_a)^2 - \sum_{a, b = 0}^{d} \left(\alpha_a\alpha_b \lambda_a \lambda_b\right)^2\right)\right) = O(\sigma^2)
\end{align}
(noting that $\lambda_a (1-\lambda_a) \leq 1$ for all $\lambda_a$ and that $\sum_a \alpha_a^2 = 1$). Then, \cref{algorithm:quantum_mean_estimation} with $t = O(\sigma^{-2})$ and $k = \Theta(\sigma^2/\epsilon^2)$ produces an estimate of $\Tr[\avM \rho]$ with the correct mean and variance
\begin{align}
    O((\epsilon^2/\sigma^2) \sigma^2)) = O(\epsilon^2)\,.
\end{align}
This estimate has standard deviation $O(\epsilon)$ and therefore is within $\epsilon$ of the $\Tr[\avM \rho]$ with constant probability, as desired. 
\end{proof}

We can gain intuition about \cref{algorithm:quantum_mean_estimation} by comparing its performance to the performance of the analogous classical algorithm, in which the average success probability of $m$ events on an unknown distribution $X$ is computed by taking $t$ samples from $X$ and counting the total number of events that succeed. Interesting, this analysis reveals possible \emph{advantages} to mean estimation in the quantum setting. 

\begin{ex} \label{eq:mean_estimation_Z_X}
Consider estimating the average accepting probability of measurements $M_1 = \ketbra{1}$ and $M_2 = \ketbra{+}$ on an unknown one qubit state. We have 
\begin{align}
    \avM = \frac{1}{2} (M_1 + M_2) = \begin{pmatrix}
    3/4 & 1/4 \\
    1/4 & 1/4 
    \end{pmatrix}
\end{align}
which has eigenvalues $\lambda_1 = (2 + \sqrt{2})/4$ and $\lambda_2 = (2 - \sqrt{2})/4$. Then, by \Cref{rmk:mean_estimation_variance}, the variance of the estimate provided by \cref{algorithm:quantum_mean_estimation} in the $t = \infty$ limit using a single copy of the unknown state $\rho$ is upper bounded by
\begin{align}
\frac{1}{4} \left(\lambda_1 - \lambda_2 \right)^2 = \frac{1}{8}\,.    
\end{align}
Interestingly, this is the same as the worst case variance in the classical case where we are trying to estimate the average success probability of two events using a single sample from an unknown distribution, \emph{and we are given the promise that the events are independent}. (To see this, note the worst case variance when estimating the success probability of a single event is $1/4$, occurring when the event happens with probability $1/2$. The worse case variance when estimating the average success probability of two independent events is then $(1/2)(1/4) = 1/8$.)
\end{ex}

On one hand, \Cref{eq:mean_estimation_Z_X} is perhaps not that surprising: quantum mechanics forces the $M_1$ and $M_2$ success probabilities to be independent, and Algorithm 5 takes advantage of that fact. On the other hand, making either an $M_1$ or $M_2$ measurement directly destroys all information about the success probability of the other, so it is perhaps surprising that \cref{algorithm:quantum_mean_estimation} extracts information about the success probabilities of both measurements using just a single copy of the unknown quantum state. 

We can also show that our algorithm is optimal up to a constant factor.  By Ref.~\cite{o2015quantum}, it requires $\Omega(d/\epsilon^2)$ copies of a state $\rho$ to test whether a state is at least $\epsilon$ away from the maximally mixed state in trace distance, and when $d$ is a constant (say $2$), it requires $\Omega(1/\epsilon^2)$ copies.  For any $2$-dimensional quantum state, we can estimate the energy of the Hamiltonians $\lambda X$, $\lambda Y$ and $\lambda Z$ (for any $\lambda \geq 0$) on the unknown state $\rho$ to precision $\epsilon/2\lambda$ to reconstruct the classical description of $\rho$ to trace distance $\epsilon$.  This allows us to distinguish it from the maximally mixed state.  Since the number of copies required is $\Omega(1/\epsilon^2)$, the number of copies required to estimate $\Tr(H\rho)$ to precision $\epsilon' = \epsilon/2\lambda$ must be at least $\Omega(1/\epsilon^2) = \Omega(\lambda^2 / \epsilon'^2)$, which matches our upper bound.  






\section{Open Problems}


We now discuss a few possible ways in which the random gentle measurement lemma presented in this paper could be strengthened, along with a more abstract questions concerning the sample complexity of Quantum Event Learning Problems in general. 

\begin{enumerate}
    \item The Gentle Random Measurement Lemma (\Cref{thm:random_gentle_total_accept}) only applies to randomly ordered two outcome projective measurements. Can it be generalized to randomly ordered two outcome POVMs? What about measurement with more than two outcomes?
    \item The Gentle Random Measurement Lemma proved in this paper bounds the expected disturbance caused by a randomly ordered sequence of measurements in terms of the probability -- over both orderings and quantum randomness -- that at least one measurement accepts. A stronger statement is possible. Is it true that, with high probability over orderings, the disturbance caused by \textit{any} randomly fixed sequence of measurements is bounded by the probability that any measurement in that sequence accepts? (i.e. is it true that with high probability a randomly ordered sequence is not an anti-Zeno sequence?)

    \item The constants appearing in the analysis of the Quantum Even Finding protocols (\Cref{thm:blended_event_finding,thm:random_event_finding}) and Threshold Search protocols (\Cref{thm:blended_measurement_threshold_search,thm:random_measurement_threshold_search}) are likely far from optimal, especially in the random measurements case. Can either a modified protocol or more sophisticated analysis improve these bounds? 
    
    \item In this paper we studied three natural Quantum Event Learning problems: Quantum OR, Mean Estimation,  Quantum Event Finding, and Quantum Threshold Search. For the first three of these problems we gave protocols, based on random or blended measurements, which showed the sample complexity of these problems was close to the sample complexity of the analogous classical problems, and for the last we matched the best known algorithm. In particular, the sample complexity of all of these Quantum Event Learning Problems is independent of the dimension of the Hilbert space in which the unknown quantum state lives.

    Yet the best upper bounds on the sample complexity of Shadow Tomography do have a dimension dependence. And, in \Cref{sec:old_quantum_event_finding}, we discuss a problem for which it appears neither randomly ordered nor blended measurements can reproduce the behavior of classical measurements. Can we give a clear delineation between Quantum Event Learning problems with and without a dimension dependence? Which of these problems can be solved efficiently with random measurements? Or perhaps (as originally asked in \cite{aaronson2019shadow}) is the true sample complexity of Shadow Tomography independent of the dimension of the underlying Hilbert space? 

\end{enumerate}







\bibliographystyle{abbrv}
\bibliography{ref.bib}

\appendix 

\section{Miscellaneous Proofs}

This appendix proves a few statements which we viewed as ``natural'' enough that a proof in the main text wasn't required, but which still required formal proof.

\begin{lem}\label{lem:repeated_measurements_increase_prob}
For any state $\rho$ and PSD matrix $M \leq 1$ with we have 
\begin{align}
    \Tr[M \rho] \leq \Tr[M^2 \rho]/\Tr[M\rho]\,. \label{eq:repeated_measurement_increase}
\end{align}
\end{lem}

\begin{proof}
Applying Cauchy-Schwarz with the Hilbert-Schmidt inner product gives 
\begin{align}
    \Tr[M\rho]^2 &= \Tr[\rho^{1/2} \rho^{1/2}M]^2 
    \leq \Tr[\rho]\Tr[M\rho M] = \Tr[M^2\rho]\,,
\end{align}
and \Cref{eq:repeated_measurement_increase} follows. 
\end{proof}
\begin{cor}\label{cor:many_repeated_measurements_increase_prob}
For any state $\rho$ and PSD matrix $M \leq 1$ with we have 
\begin{align}
    \Tr[M \rho] \leq \Tr[M^{k+1} \rho]/\Tr[M^{k}\rho]\,.
\end{align}
\end{cor}
\begin{proof}
\Cref{lem:repeated_measurements_increase_prob} and the observation that $M^{(k-1)/2}\rho M^{(k-1)/2}/\Tr[M^{k-1} \rho]$ is a valid state implies that
\begin{align}
    \Tr[M^k \rho]/\Tr[M^{k-1} \rho] &= \Tr[M \left(\frac{M^{(k-1)/2} \rho M^{(k-1)/2}}{\Tr[M^{k-1}\rho]} \right)] \\
    &\leq \Tr[M^2 \left(\frac{M^{(k-1)/2} \rho M^{(k-1)/2}}{\Tr[M^{k-1}\rho]} \right)] / \Tr[M \left(\frac{M^{(k-1)/2} \rho M^{(k-1)/2}}{\Tr[M^{k-1}\rho]} \right)] \\
    &= \Tr[M^{k+1}\rho]/\Tr[M^k \rho]\,.
\end{align}
Applying this result inductively proves the corollary. 
\end{proof}

\begin{lem}
\label{lem:variance_upper_bound}
Let $\{\lambda_1, \lambda_2, \ldots, \lambda_m\} \in \mathbb{R}$ be an arbitrary set of numbers, with $\lambda_1 \geq \lambda_2 \geq \ldots \geq \lambda_m$. Let $X$ be a random variable that takes value $\lambda_i$ with probability $p_i$. Then 
\begin{align}
    \Var[X] \leq \frac{(\lambda_1 - \lambda_m)^2}{4}\,,
\end{align}
where the bound is saturated when $p_1 = p_m = \frac{1}{2}$ and all other values of $p_i = 0$.
\end{lem}

\begin{proof}
We first show the variance is maximized when $p_2 = p_3 = \ldots = p_{m-1} = 0$. Assume for contradiction that this is not true, and that $\Var[X]$ is maximized when $p_i \neq 0$ for some $i \in \{2,\ldots,m-1\}$. Then, using the law of total variance, we can write 
\begin{align}
    \Var[X] &= \Var[X | X \neq \lambda_i] + \Var[X | X = \lambda_i] + \Var[\{\mathbb{E}[X | X = \lambda_i], \mathbb{E}[X | X \neq \lambda_i]\}] \label{eq:lotv_discrete_1_entry}\\
    &= \Var[X | X \neq \lambda_i] + p_i(1-p_i)(\lambda_i - \{\mathbb{E}[X | X \neq \lambda_i])^2
\end{align}
where we noted that $\Var[X | X = \lambda_i] = 0$ and that the the last term in \Cref{eq:lotv_discrete_1_entry} was just the variance of a rescaled Bernoulli distribution. But then we see 
\begin{align}
    \Var[X] &< \max_{j \in \{1, m\}} \left(\Var[X | X \neq \lambda_i] + p_i(1-p_i)(\lambda_j - \{\mathbb{E}[X | X \neq \lambda_i])^2 \right)
\end{align}
which is the variance of the distribution where the value of $\lambda_i$ is set equal to $\lambda_j$ (or equivalently, probability $p_i$ is removed from the probability of outcome $\lambda_i$ and added to the probability of outcome~$\lambda_j$). This contradicts the assumption that $X$ was the distribution maximizing $\Var[X]$ and we conclude that $\Var[X]$ is maximized when $p_2 = p_3 = \ldots = p_{m-1} = 0$.

From here, the result follows quickly from the observation that the distribution $X$ with $p_2 = p_3 = \ldots = p_{m-1} = 0$ is a shifted and rescaled Bernoulli distribution. 
\end{proof}

\section{Limitations of Random and Blended Measurements}
\label{sec:old_quantum_event_finding}

So far in this paper we have discussed ways in which gentle and blended measurements applied to an unknown state can reproduce the behavior of classical measurements applied to an unknown probability distribution. In this appendix we discuss a situation in which it appears that gentle and blended measurements \textit{cannot} reproduce the behavior of classical measurements. 

To motivate this situation we first discuss the classical union bound, which states that the probability of any event from a set $\cE = \{E_1, E_2, \ldots, E_m\}$ of accepting on a sample drawn from some unknown probability distribution $X$ (equivalently, one minus the probability that all the events reject) is upper bounded by the sum of the probabilities of each individual event accepting:
\begin{align}
     \p_X[E_1 \vee E_2 \vee \ldots \vee E_m] =  1 - \p_X[\neg E_1 \wedge \neg E_2 \wedge \ldots \wedge \neg E_m] &\leq \sum_{E_i \in \cE} \p_X[E_i]\,.
\end{align}
Quantum analogs of this statement are called \textit{Quantum Union Bounds}~\cite{sen2012achieving,gao2015quantum,khabbazi2019union}. As an example, we state the union bound proven in~\cite{gao2015quantum} (though we note a stronger version is given in~\cite{khabbazi2019union}). 

\begin{thm}[Gao's Quantum Union Bound] 
For any sequence of two outcome projective measurements $(A_1, A_2, \ldots, A_m)$ and any quantum state $\rho$ we have: 
\begin{align}
    1 - \Tr[(1-A_m)\ldots(1-A_2)(1-A_1)\rho(1-A_1)(1-A_2)\ldots(1-A_m)] \leq 4 \sum_i \Tr[A_i \rho]\,.
\end{align}
\end{thm}

The question we consider here is whether any quantum union bound can be generalized to apply to just a subset of the measurements applied to a quantum system. Classically (because classical measurements don't disturb the system on which they act) this generalization is immediate: given a set of events $\cE = \{E_1, E_2, \ldots, E_m\}$ an any subset $\cF \subseteq \cE$ the probably any event from $\cF$ accepts is still bounded: 
\begin{align}
    \p_X[\bigvee_{F \in \cF} F] \leq \sum_{F \in \cF} \p_X[F]\,.
\end{align}
In the quantum case measurements can damage the state on which they act, and a direct generalization of the union bound to subsets of measurements seems unlikely.\footnote{To see why, consider the set of measurements $\cM = \{\ketbra{1}, \ketbra{+}\}$ and state $\rho = \ketbra{-}$. The $\ketbra{+}$ measurement initially has $0$ probability of accepting on a quantum system in state $\rho$ but, after the $\ketbra{1}$ measurement is applied to the system, the $\ketbra{+}$ measurement accepts with probability $1/2$.} Yet even weak generalizations of the quantum union bound are interesting to consider. We formalize one such possible generalization in the next definition. 
\begin{defn} \label{defn:Strong_Quantum_Union_Bound}
Let $\cM = \{M_1, M_2, \ldots, M_m\}$ be a set of two outcome measurements, $\rho$ be an unknown quantum state, and fix $\epsilon > 1/2$. Then a protocol with sample access to $\rho$ respects the \emph{Subset Quantum Union Bound} if:
\begin{enumerate}
    \item It returns some measurement $M_i \in \cM$ with constant probability provided there exists a measurement $M_j \in \cM$ with
    \begin{align}
        \Tr[M_j \rho] > \epsilon\,. 
    \end{align}
    \item For any subset of measurements $\cA \subseteq \cM$ and constant $0 \leq \gamma \leq 1$ satisfying 
    \begin{align}
    \sum_{A \in \cA}\Tr[A \rho] < \gamma\,,  
    \end{align}
    the probability that the protocol returns any measurement $A \in \cA$ is bounded above by $C(\gamma)$,
    where $C(\gamma)$ is some continuous function, independent of $m$, with $C(0) < 1$. 
\end{enumerate}
\end{defn} 

Note that if measurements $M_1, M_2, \ldots, M_m$ all commute with each other (i.e. in the classical case) a protocol which respects the subset quantum union bound is just to apply all measurements to a quantum system in state $\rho$ and then returning a random measurement which accepts. The classical union bound then guarantees that the second condition holds with $C(\gamma) = \gamma$. The question raised by \Cref{defn:Strong_Quantum_Union_Bound} is whether such behavior can be reproduced approximately when measurements do not commute. 

We next give an example and a (numerically verified) conjecture which suggest that random and blended measurements are unlikely to give protocols which satisfy the subset quantum union bound, at least when applied naively.

\begin{ex} \label{eq:blended_event_finding_counterexample}
For any constants $\epsilon, \delta > 0$ there is a set of two outcome measurements $\cM = \{M_1, M_2, \ldots, M_m\}$, state $\rho$, and subset of measurements $\cS \subseteq \cM$ with the following properties:
\begin{enumerate}
    \item If a the blended measurement $\blend(\cM)$ is made $m$ times in sequence on a quantum in system it accepts with probability $1 - \delta$.
    \item The first measurement to accept is a measurement $M_i \in \cS$ with probability $1 - \epsilon$.
    \item $\sum_{M_j \in \cS} \Tr[M_j \rho]  = 0.$
    \
\end{enumerate}
\end{ex}

\begin{proof}
This example is based on a 3 outcome measurement,\footnote{Found thanks to Luke Schaeffer.} with measurement operators:
\begin{align}
    E_1 &= \left( \frac{1 + \epsilon - \epsilon^3 - \epsilon^4}{1+ \epsilon}  \ketbra{1}\right)^{1/2} \\
    E_2 &= \left( \frac{\epsilon}{1 + \epsilon}  \big(\ket{0} + \epsilon \ket{1} \big) \big( \bra{0} + \epsilon \bra{1} \big) \right)^{1/2} \\
    E_3 &= \left(\frac{1}{1+\epsilon} \big( \ket{0} + \epsilon^2 \ket{1} \big) \big( \bra{0} + \epsilon^2 \bra{1} \big)\right)^{1/2}\,.
\end{align}
We view $E_1$ as the rejecting outcome, and outcomes $E_2$ and $E_3$ as accepting outcomes. Direct calculation shows that if this measurement is applied $k$ times to a quantum system initially in state $\ket{\psi} \propto \epsilon \ket{0} - \ket{1}$ at least one accepting outcome is observed with probability
\begin{align}
    1 - \Tr[(E_1)^{2k} \rho] \geq 1 - \left(1 - \frac{\epsilon^3}{1+\epsilon}\right)^k
\end{align}
and so if $k = \omega(\epsilon^{-3})$ at least one accepting outcome is observed with high probability. Additionally, on all measurements after the first blended measurement rejects we see the post measurement state is proportional to $\ketbra{1}$ and 
\begin{align}
    \Tr[(E_2)^2 \ketbra{1}] = \frac{\epsilon^3}{1+\epsilon} = \epsilon \Tr[(\ketbra{E_3}^3 \ketbra{1}]\,.
\end{align}
So we see that for all blended measurements applied after the first blended measurement rejects
\begin{align}
     \mathbb{P}[\text{Outcome } E_3 \text{ is observed}] = \epsilon \mathbb{P}[\text{Outcome } E_2 \text{ is observed}]\,.
\end{align}
Then we see the first accepting measurement observed corresponds to outcome $E_2$ with probability approximately $1-\epsilon$, despite the fact $\Tr[(M_2)^2 \ketbra{\psi}]= 0$. 

To turn this into a blended measurement we define the sets of measurements
\begin{align}
    \mathcal{A} &= \{\left\lceil \epsilon^{-3} \right\rceil \text{ copies of the projector } (1+\epsilon^4)^{-1}\big( \ket{0} + \epsilon^2 \ket{1} \big) \big( \bra{0} + \epsilon^2 \bra{1} \big) \} \\
    \mathcal{B} &= \{\left\lceil \epsilon^{-2} \right\rceil \text{ copies of the projector } (1+\epsilon^2)^{-1}\big( \ket{0} + \epsilon \ket{1} \big) \big( \bra{0} + \epsilon \bra{1} \big) \}\,.
\end{align}
The same argument as above shows that the blended measurement $\blend(\cA \cup \cB)$ applied $\abs{\cA \cup \cB}$ times to the state $\ket{\psi}$ is overwhelmingly likely to accept on a measurement in $\cB$, producing the desired counterexample. 
\end{proof}

\begin{con}
There exists sets of two outcome measurements $\cA$ and $\cB$ and state $\rho$ with the following properties:
\begin{enumerate}
    \item If $\abs{\cA \cup \cB}$ measurements are selected at random (with replacement) from the set $\cA \cup \cB$ and applied in sequence to a quantum in system initially in state $\rho$ a measurement accepts with probability at least $1-\delta$. 
    \item The first accepting measurement is a measurement in the set $\cB$ with probability at least $1 - \epsilon$.
    \item $\sum_{M_j \in \cB} \Tr[M_j \rho] = 0.$
\end{enumerate}
\end{con}

\begin{proof}[Evidence] We construct a set of measurements $\cA$ and $\cB$ which we expect will satisfy the above conjecture. The measurements $\cA$ and $\cB$ are based on a modified version of the measurements constructed in \Cref{eq:blended_event_finding_counterexample}. First, define measurement operators $E_1, E_2, E_3$ as in that example. Then, building on them, define two distinct two outcome measurements: 
\begin{align}
    M_B &= (E_2)^2 \\
    M_A &= I - \left((I - (E_2)^{2})^{-1/4}(E_1)(I - (E_2)^{2})^{-1/4}\right)^2
\end{align}
(Recall that ``the two outcome measurement $M$'' is the measurement with measurement operators $\{\sqrt{M}, \sqrt{1-M}\}$). Next, define the set $\cA$ to contain $\omega(\epsilon^{-3})$ copies of $M_A$, and the set $\cB$ to contain the same number of copies of $M_B$. Let $\ket{\psi}$ be the same state as defined in \Cref{eq:blended_event_finding_counterexample}. 

To motivate this choice of measurements we note the probability that measurement $M_B$ accepts after measurements $M_A$ and $M_B$ are alternated $j$ times on a state $\rho$ and all reject is given by
\begin{align}
    &\Tr[M_B \left((I-M_A)^{1/2} (I-M_B)^{1/2}\right)^k \rho \left((1-M_B)^{1/2} (I-M_A)^{1/2}\right)^j ]\\
    &\hspace{20pt}\geq \Tr[M_B (1 - M_B)^{1/2} \left((I-M_A)^{1/2} (I-M_B)^{1/2}\right)^j \rho \left((1-M_B)^{1/2} (I-M_A)^{1/2}\right)^j ]\\
    &\hspace{20pt}= \Tr \bigg[ M_B \left((I-M_B)^{1/4} (I-M_A)^{1/2} (I-M_B)^{1/4}\right)^j \\
    &\hspace{80pt}(I-M_B)^{1/4} \rho (1-M_B)^{1/4} \left((1-M_B)^{1/4} (I-M_A)^{1/2} (I-M_B)^{1/4}\right)^j \bigg]\\
    &\hspace{20pt}= \Tr[E_2^2 (E_1)^j (1- (E_2)^2)^{1/4} \rho (1- (E_2)^2)^{1/4}  (E_1)^j]\,. 
\end{align}
Then, if measurements $M_A$ and $M_B$ are alternated $k$ times in total on the initial state $\rho = \ketbra{\psi}$ with $\ket{\psi} = \epsilon\ket{0} - \ket{1}$ we see the overall probability of measurement $M_B$ accepting is lower bounded by
\begin{align}
    &\Tr[(E_2)^2 \rho] + \sum_{j=1}^{k-1} \Tr[(E_2)^2 (E_1)^j (1- (E_2)^2)^{1/4} \rho (1- (E_2)^2)^{1/4}  (E_1)^j] \\
    &\hspace{220pt}= \sum_{j=0}^{k-1} \Tr[(E_2)^2 (E_1)^j \rho (E_1)^j]
\end{align}
(where the equality follows because $\Tr[\rho E_2] = 0$). But this is exactly the probability that outcome $E_2$ if the first accepting outcome observed during $k$ applications of the blended measurement discussed in the proof of \Cref{eq:blended_event_finding_counterexample}. It follows that if measurements $M_A$ and $M_B$ are alternated enough times on the state $\rho$ measurement $M_B$ is very likely to be the first accepting measurement, despite the fact that $\Tr[M_B \rho] = \Tr[(E_2)^2 \rho] = 0$. 

We do not have an algebraic analysis of the situation where measurements $M_A$ and $M_B$ are drawn at random (with replacement) from the set $\cA \cup \cB$ and applied to the state $\rho$. However, numerical tests of this system with $\epsilon = 0.02$ and $\abs{\cA} = 20 \epsilon^{-3}$ show that a measurement from the set $\cB$ is the first to accept with probability $> 0.99$.\footnote{We cannot check all orderings of measurements from $\cA \cup \cB$, but test the parameters for randomly chosen orderings. They appear stable over orderings, but we also do not have a formal proof of this fact. }
\end{proof}
\end{document}

%% file: packages.tex
\usepackage{amsmath}
\usepackage{amssymb}
\usepackage{amsthm}
\usepackage{amsfonts}
\usepackage[pagebackref]{hyperref}
\hypersetup{
    colorlinks=true, 
    linkcolor=blue, 
    citecolor=blue, 
    urlcolor=blue 
}
\usepackage{microtype}

\usepackage{accents}
\usepackage[toc,page]{appendix}
\usepackage{array}
\usepackage[utf8]{inputenc}
\usepackage[main=english]{babel}
\usepackage{braket}
\usepackage{booktabs}
\usepackage{bigdelim}
\usepackage{cancel}
\usepackage[shortlabels]{enumitem}
\usepackage{float}
\usepackage[bottom]{footmisc}
\usepackage[margin=1in]{geometry}
\usepackage{graphicx}
\usepackage[capitalize,nameinlink]{cleveref} 
\usepackage{mathtools}
\usepackage{multirow}
\usepackage{physics}
\usepackage{tabularx}
\usepackage[dvipsnames]{xcolor}
\usepackage{tikz}
\usepackage[textsize=footnotesize]{todonotes}
\usepackage{qcircuit}
\usetikzlibrary{positioning}
\usetikzlibrary{decorations.pathreplacing}
\usetikzlibrary{cd}
\usetikzlibrary{arrows.meta}
\usetikzlibrary{calc}
\usepackage{bbold}

\usepackage{tikzsymbols}

\usepackage{xspace}

\renewcommand{\backref}[1]{}

\renewcommand{\backrefalt}[4]{%
\ifcase #1 %
\or
[p.\ #2]%
\else
[pp.\ #2]%
\fi}

\makeatletter
\renewcommand{\paragraph}{%
 \@startsection{paragraph}{4}%
 {\z@}{2.25ex \@plus .5ex \@minus .3ex}{-1em}%
 {\normalfont\normalsize\bfseries}%
}
\makeatother

\makeatletter
\newcommand{\para}{%
 \@startsection{paragraph}{4}%
 {\z@}{2ex \@plus 3.3ex \@minus .2ex}{-1em}%
 {\normalfont\normalsize\bfseries}%
}
\makeatother

%% file: Adam.tex
\newtheorem{lem}{Lemma}
\newtheorem*{lem*}{Lemma}
\newtheorem{thm}[lem]{Theorem}
\newtheorem*{thm*}{Theorem}

\newtheorem{claim}[lem]{Claim}
\newtheorem{cor}[lem]{Corollary}
\newtheorem{defn}[lem]{Definition}

\newtheorem{con}[lem]{Conjecture}
\theoremstyle{definition}

\def\beq{\begin{equation}}
\def\eeq{\end{equation}}

\def\ben{\begin{enumerate}}
\def\een{\end{enumerate}}

\def\bem{\begin{bmatrix}}
\def\eem{\end{bmatrix}}

\def\cM{{\mathcal M}}
\def\cA{{\mathcal A}}

%% file: headers.tex
\usepackage[utf8]{inputenc}
\usepackage{framed}
\usepackage{xspace}

\newcommand{\adj}{^{\dag}}
\newcommand{\ACCEPT}{\textsc{accept}\xspace}
\newcommand{\REJECT}{\textsc{reject}\xspace}
\newcommand{\Var}{\operatorname*{Var}}

\Crefname{obs}{Observation}{Observations}
\crefname{obs}{observation}{observations}
\crefname{thm}{theorem}{theorems}
\Crefname{thm}{Theorem}{Theorems}
\Crefname{secinapp}{Appendix}{Appendices}
\newtheorem{rmk}[lem]{Remark}
\newtheorem*{rmk*}{Remark}
\newtheorem{ex}[lem]{Example}
\newtheorem{algorithm}{Algorithm}

\newcommand{\accept}[1]{\operatorname{Accept}(#1)}
\newcommand{\reject}[1]{\operatorname{Reject}(#1)}
\newcommand{\baccept}[1]{\operatorname{Accept}_\cB(#1)}
\newcommand{\bmaccept}[1]{\operatorname{Accept}_{\cB(\cM)}(#1)}
\newcommand{\breject}[1]{\operatorname{Reject}_\cB(#1)}
\newcommand{\state}[1]{\rho^{(#1)}}
\newcommand{\bstate}[1]{\rho_{\cB}^{(#1)}}
\newcommand{\bmstate}[1]{\rho_{\cB(\cM)}^{(#1)}}
\newcommand{\mm}[2]{\mathcal{T}_{#1}^{(#2)}} 

\newcommand{\cS}{\mathcal{S}}

\newcommand{\cE}{\mathcal{E}}
\newcommand{\cF}{\mathcal{F}}
\newcommand{\blend}{\mathcal{B}}
\newcommand{\cB}{\mathcal{B}}

\newcommand{\expec}{\mathbb{E}}
\newcommand{\notcM}{\overline{\cM}}
\newcommand{\avM}{\mathbf{M}}
\newcommand{\avnotM}{\mathbf{\overline{M}}}

\newcommand{\pupb}{p_{\uparrow}}
\newcommand{\plob}{p_{\downarrow}}
\newcommand{\paccept}{p_{accept}}
\newcommand{\pbaccept}{p^{\cB}_{accept}}
\newcommand{\preject}{p_{reject}}
\newcommand{\rhoreject}[1]{\rho^{(#1)}}
\newcommand{\rhobreject}[1]{\rho_{\cB}^{(#1)}}

\newcommand{\fidelity}[2]{\mathrm{F}\left(#1, #2\right)}
\newcommand{\wtgamma}{\widetilde{\gamma}}

\newcommand{\smallWeight}{\beta}
\newcommand{\p}{\mathbb{P}}
\newcommand{\firstAcceptB}[1]{\operatorname{Return}_\cB(#1)}
\newcommand{\goodFirstAcceptB}[1]{\operatorname{Success}_\cB(#1)}
\newcommand{\firstAccept}[1]{\operatorname{Return}(#1)}
\newcommand{\goodFirstAccept}[1]{\operatorname{Success}(#1)}
\newcommand{\system}[1]{\mathsf{#1}}

\makeatletter
\newtheorem*{rep@theorem}{\rep@title}
\newcommand{\newreptheorem}[2]{%
\newenvironment{rep#1}[1]{%
 \def\rep@title{#2 \ref{##1}}%
 \begin{rep@theorem}}%
 {\end{rep@theorem}}}
\makeatother

\newreptheorem{theorem}{Theorem}
\crefname{question}{question}{questions}
\Crefname{question}{Question}{Questions}
\crefname{thm}{theorem}{theorems}
\Crefname{thm}{Theorem}{Theorems}

\newcommand{\Binomial}{\mathsf{Binom}}
\newcommand{\bin}[3]{\mathsf{B}(#1, #2, #3)}